\documentclass[12pt, onecolumn]{IEEEtran}

\usepackage{setspace}
\usepackage[a4paper]{geometry}
\usepackage{amsmath}
\usepackage{amsthm}
\usepackage{amssymb}
\usepackage{verbatim}
\usepackage{algorithm}
\usepackage{algorithmic}
\usepackage{graphicx}
\usepackage{paralist}
\usepackage{color}

\ifCLASSOPTIONcompsoc
\usepackage[tight,normalsize,sf,SF]{subfigure}
\else
\usepackage[tight,footnotesize]{subfigure}
\fi

\theoremstyle{plain}
\newtheorem{lemma}{Lemma}
\newtheorem{theorem}{Theorem}
\newtheorem{definition}{Definition}

\theoremstyle{definition}

\floatname{algorithm}{Algorithm}
\newcommand{\argmax}{\operatornamewithlimits{argmax}}
\newcommand{\argmin}{\operatornamewithlimits{argmin}}
\allowdisplaybreaks[4]

\begin{document}

\newlength{\figurewidth}\setlength{\figurewidth}{0.6\columnwidth}


\addtolength{\topmargin}{-0.5\baselineskip}
\addtolength{\textheight}{\baselineskip}

\title{Let Cognitive Radios Imitate: Imitation-based Spectrum Access for Cognitive Radio Networks}

\newcounter{one}
\setcounter{one}{1}
\newcounter{two}
\setcounter{two}{2}

\author{Stefano Iellamo\footnote{S. Iellamo and M. Coupechoux are with the department of Computer Science and Networks of Telecom ParisTech - LTCI CNRS 5141, 46 Rue Barrault, Paris 75013, France (e-mail: \{iellamo,coupecho\}@enst.fr).} \qquad\qquad Lin Chen\footnote{L. Chen is with Laboratoire de Recherche en Informatique (LRI), the department of Computer Science of the University of Paris-Sud XI, 91405 Orsay, France (e-mail: chen@lri.fr).} \qquad\qquad Marceau Coupechoux}

\addtolength{\floatsep}{-\baselineskip}
\addtolength{\dblfloatsep}{-\baselineskip}
\addtolength{\textfloatsep}{-\baselineskip}
\addtolength{\dbltextfloatsep}{-\baselineskip}
\addtolength{\abovedisplayskip}{-1ex}
\addtolength{\belowdisplayskip}{-1ex}
\addtolength{\abovedisplayshortskip}{-1ex}
\addtolength{\belowdisplayshortskip}{-1ex}

\maketitle

\vspace{-1cm}

\begin{abstract}
In this paper, we tackle the problem of opportunistic spectrum access in large-scale cognitive radio networks, where the unlicensed Secondary Users (SU) access the
frequency channels partially occupied by the licensed Primary Users (PU). Each channel is characterized by an availability probability unknown to the SUs. We apply evolutionary game theory to model
the spectrum access problem and develop distributed spectrum access policies based on imitation, a behavior rule widely applied in human societies consisting of imitating successful behavior.
We first develop two imitation-based spectrum access policies based on the basic Proportional Imitation (PI) rule and the more advanced Double Imitation (DI) rule given that a SU can imitate any other SUs. We then
adapt the proposed policies to a more practical scenario where a SU can only imitate the
other SUs operating on the same channel. A systematic theoretical analysis is presented for both scenarios on the induced imitation dynamics and the convergence properties of the proposed policies to an imitation-stable equilibrium, which is also the $\epsilon$-optimum of the system.
Simple, natural and incentive-compatible, the proposed imitation-based spectrum access policies can be implemented distributedly
based on solely local interactions and thus is especially
suited in decentralized adaptive learning environments
as cognitive radio networks.
\end{abstract}




\section{Introduction}
\label{sec:intro}

\textit{Cognitive radio} \cite{Haykin05}, with its capability to flexibly configure its transmission parameters, has emerged in recent years as a promising paradigm to enable more efficient spectrum utilization. Spectrum access models in cognitive radio networks can be classified into three categories, namely exclusive use (or operator sharing), commons and shared use of primary licensed spectrum~\cite{Bud07}. In the last model, unlicensed secondary users (SU) are allowed to access the spectrum of licensed primary users (PU) in an opportunistic way. In this case, a well-designed spectrum access mechanism is crucial to achieve efficient spectrum usage.

In this paper, we focus on the generic model of cognitive networks consisting of multiple frequency channels, each characterized by a channel availability probability determined by the activity of PUs on it. In such model, from the individual SU's perspective, a challenging problem is to compete (or coordinate) with other SUs in order to opportunistically access the unused spectrum of PUs to maximize its own payoff (e.g., throughput); at the system level, a crucial research issue is to design efficient spectrum access protocols achieving optimal spectrum usage.

We tackle the spectrum access problem in large-scale cognitive radio networks from an evolutionary game theoretic angle. We formulate the spectrum access problem as a non-cooperative game and develop distributed spectrum access policies based on imitation, a behavior rule widely applied in human societies consisting of imitating successful behavior. We establish the convergence of the proposed policies to an imitation-stable equilibrium which is also the $\epsilon$-optimum of the system. Simple, natural and incentive-compatible mechanism, the proposed spectrum access policies can be implemented distributedly based on solely local interactions and thus is especially suited in decentralized adaptive learning environments as cognitive radio networks.

The motivation of applying evolutionary game theory and imitation-based strategy in the study of the spectrum access problem is tri-fold.
\begin{itemize}
\item First, evolutionary game theory is a powerful tool to study the interaction among players and the system dynamic in terms of population. Stemmed from classic game theory and Darwin's evolution theory, it can explicitly capture the fundamental relationship among \textit{competition}, \textit{cooperation} and \textit{communication}, three crucial elements in the design of any spectrum access protocols in cognitive radio networks.
\item Second, compared with replicator dynamic, the most explored evolutionary model which mimics the effect of natural selection, imitation dynamic captures the spreading of successful strategies through imitation ra\-ther than inheritance, which is more adapted in games played by autonomous decision makers as in our case.
\item Third, evolutionary game theory, especially imitation dynamic which relies solely on local interactions, provides a theoretic tool for the design of distributed channel access protocols based on local information which is particularly suited in decentralized environments as cognitive radio networks.
\end{itemize}

In our analysis, we start by developing the imitation-based spectrum access policies where a SU can imitate any other SUs. More specifically, we develop two spectrum access policies based on the following two imitation rules:
\begin{itemize}
\item the Proportional Imitation (PI) rule where a SU can sample one other SU;
\item the more advanced adjusted proportional imitation rule with double sampling (Double Imitation, DI) where a SU can sample two other SUs.
\end{itemize}
Under both imitation rules, each SU strives to improve its individual payoff by imitating other SUs with higher payoff. We then adapt the proposed spectrum access policies to a more practical scenario where a SU can only imitate the other SUs operating on the same channel. A systematic theoretical analysis is presented for both scenarios on the induced imitation dynamics and the convergence properties of the proposed policies to an imitation-stable equilibrium, which is also the $\epsilon$-optimum of the system.

The key contribution of our work in this paper lies in the systematical application of the natural imitation behavior to address the spectrum access problem in cognitive radio networks, the design of a distributed imitation-based channel access policy, and the theoretic analysis on the induced imitation dynamic and the convergence to an efficient and stable system equilibrium.

The rest of the paper is structured as follows. Section~\ref{sec:models} presents the system model followed by the formulation of the spectrum access game. Section~\ref{sec:imitation} describes the proposed imitation-based spectrum access policies in the scenario where a SU can imitate any other SUs. In Section~\ref{sec:NewScenario}, we adapt the proposed policies to the scenario where a SU can only imitate the other SUs operating on the same channel. Section~\ref{sec:simu} presents simulation results on the performance of the proposed policies. Section~\ref{sec:related_work} discusses related work in the literature. Section~\ref{sec:conclusion} concludes the paper.

\section{System Model and Spectrum Access Game Formulation}
\label{sec:models}
In this section, we present the system model of our work with the notations used, followed by the game formulation of the spectrum access problem, which serves as the basis of the analysis in subsequent sections.

\subsection{System Model}
We consider a primary network consisting of a set $\cal C$ of $C$ frequency channels, each with bandwidth $B$\footnote{Our analysis can be extended to study the heterogeneous case with different channel capacities.}. The users in the primary network are
operated in a synchronous time-slotted fashion. A set ${\cal N}$ of $N$ SUs tries to opportunistically access the channels when they are left free by PUs. Let $Z_i(k)$ be the random variable equal to $1$ when of channel $i$ is unoccupied by any PU at slot $k$ and $0$ otherwise. We assume that the process $\{Z_i(k)\}$ is stationary and independent for each $i$ and $k$. We also assume
that at each time slot, channel $i$ is free with probability $\mu_i$, i.e., $\mathbb{E}[Z_i(k)]=\mu_i$. The channel availability probabilities $\mu \triangleq \{\mu_i$\} are {\it a priori} not known by SUs. We assume perfect sensing at the SUs, i.e., any transmission of any PU on a channel is perfectly sensed by SUs sensing that channel and thus no collision occurs between PUs and SUs.

In our work, each SU $j$ is modelled as a rational decision maker, striking to maximize the throughput it can achieve, denoted as $T_j$, which can be expressed as a function of $\mu_i$ and $n_{s_j}$, where $s_j$ denotes the channel which $j$ chooses, $n_{s_j}$ denotes the number of SUs on channel $s_j$. More formally, the expected value of $T_j$ can be written as:
\begin{eqnarray*}
\mathbb{E}[T_j]=f(\mu_i, n_{s_j}).
\end{eqnarray*}
In order to perform a closed-form analysis, we focus on the scenario where the channel capacity is evenly shared among all SUs on the channel when it is free, i.e.,
\begin{eqnarray*}
\mathbb{E}[T_j]=f(\mu_{s_j}, n_{s_j})=B\mu_{s_j}/n_{s_j}.
\end{eqnarray*}
It should be noted that $f(\mu_{s_j}, n_{s_j})$ depends on the MAC protocol implemented at the cognitive users. Beside the evenly shared model considered in this paper, several other models are also largely applied in practice such as the CSMA-based random access model. Our work in this paper can be adapted in those cases by defining appropriate function $f$.

\subsection{Spectrum Access Game Formulation}
\label{subsec:congestion_game}

To study the interactions among autonomous selfish SUs and to derive distributed channel access policies, we formulate the channel selection problem as a spectrum access game where the players are the SUs. Each player $j$ stays on a channel $i$ to opportunistically exploit the unused spectrum of PUs to maximize its expected throughput.
The game is defined formally as follows:

\begin{definition}
The spectrum access game $G$ is a 3-tuple ($\cal N$, $\cal C$, $\{U_j\}$), where $\cal N$ is the player set, $\cal C$
is the strategy set of each player. Each player $j$ chooses its strategy $s_j\in{\cal C}$
to maximize its normalized utility function $U_j$ defined as
\begin{eqnarray*}
U_j=\mathbb{E}[T_j]/B=\mu_{s_j}/n_{s_j}.
\end{eqnarray*}
\end{definition}



The solution of the spectrum access game $G$ is characterized by a
Nash Equilibrium (NE)~\cite{Myerson91}, a strategy profile from which no player has incentive to deviate unilaterally. Using the related theory on congestion games, we can establish the existence and the uniqueness of the NE in the spectrum access game $G$ for the asymptotic case ($N\rightarrow\infty$) in the following theorem.

\begin{theorem}
In the asymptotic case where $N$ is large, $G$ admits a unique NE. At the NE, there are $x_i^*N$ SUs staying with channel $i$, where $x_i^*=\frac{\mu_i}{\sum_{l\in {\cal C}} \mu_l}$.
\label{th:ne_congestion_game}
\end{theorem}

\begin{proof}
Given the form of SUs' utility function, it follows from in~\cite{Milchtaich96} that the spectrum access game is a congestion game. Moreover, in the asymptotic case approximating the game $G$ by a game with a continuous set of users, denote $\mathbf{x}\triangleq\{x_i, i\in{\cal C}\}$ where $x_i$ denotes the proportion of SUs choosing channel $i$, we can write the potential function of the congestion game as follows:
\begin{equation*}
P(x)\triangleq \sum_{i\in {\cal C}}\int_{\epsilon_0}^{x_iN} \frac{\mu_i}{t}dt,
\end{equation*}
where $\epsilon_0>0$ is a small constant introduced to avoid the non-integral point of $\mu_i/t$ at $0$. We can verify that for a SU $j$ staying on channel $i$, it holds that :
\begin{eqnarray*}
\frac{\partial P(x)}{\partial x_i} = \mathbb{E}[U_j(\mu_i, x_iN)].
\end{eqnarray*}

To derive the NE of $G$, we seek the maximum of $P(x)$. To this end, we develop $P(x)$ as
\begin{equation*}
P(x) = \sum_{i\in {\cal C}} \frac{\mu_i}{N} (\log x_i - \log \epsilon_0).
\end{equation*}

To find the maximum of $P(x)$, we solve the following optimization problem
\begin{eqnarray*}
\max_{\mathbf{x}} \ P(x) \quad s.t. \quad \sum_{i\in {\cal C}} x_i =1 \ \mbox{and} \ x_i>0, \forall i\in{\cal C},
\end{eqnarray*}
which has a unique solution because the KKT conditions are necessary and sufficient as $P(x)$ is concave and the constraint is linear. After some straightforward algebraic operations, we can find the unique maximum $\mathbf{x^*}\triangleq \{x_i^*\}$ as follows:
\begin{equation*}
x_i^*=\frac{\mu_i}{\sum_{l\in {\cal C}} \mu_l} \quad \forall i\in{\cal C}.
\end{equation*}
The maximum $\mathbf{x^*}$ is also the unique NE of $G$.
\end{proof}

We can observe two desirable properties of the unique NE derived in Theorem~\ref{th:ne_congestion_game}:

\begin{compactitem}
\item the NE is optimal from the system perspective as the total throughput of the network achieves its optimum at the NE;
\item the NE ensures that the spectrum resource is shared fairly among SUs.
\end{compactitem}

One critical challenge in the analyzed spectrum access game is the design of distributed spectrum access strategies for rational SUs to converge to the NE without the \textit{a priori} knowledge of $\mu$. In response to this challenge, we develop in the sequel sections of this paper an efficient spectrum access policy. Our proposed policy can be implemented distributedly based on solely local interactions without any knowledge on the channel statistics and thus is especially suited in decentralized adaptive learning environments as cognitive radio networks. In terms of performance, we demonstrate both analytically and numerically that the proposed channel access policy converges to the $\epsilon$-NE\footnote{A strategy profile is an $\epsilon$-NE if no player can gain more than $\epsilon$ in payoff by unilaterally deviating from his strategy.} of $G$ which is also the $\epsilon$-optimum of the system.

\section{Imitation-based spectrum access policies}
\label{sec:imitation}

The spectrum access policy we develop is based on \textit{imitation}. As a behavior rule widely observed in human societies, imitation captures the behavior of a rational player that mimics the actions of other players with higher payoff in order to improve its own payoff. The induced imitation dynamic models the spreading of successful strategies under imitation~\cite{Schlag96}. In this section, we focus on the scenario where a SU can imitate any other SUs and develop two spectrum access policies based on the proportional imitation rule and the double imitation rule. We analyze the induced dynamic of the imitation process and show the convergence of the proposed policy to the $\epsilon$-NE of $G$. In the next section, we extend our efforts to a more practical scenario where a SU can only imitate the other SUs operating on the same channel and develop an adapted imitation-based spectrum access policy in the new context.

\subsection{Spectrum Access Policy Based on Proportional Imitation}

Algorithm~\ref{algo:pir} presents our proposed spectrum access policy based on the proportional imitation rule, termed as PISAP. The core idea is:
at each iteration, each SU randomly selects another SU in the network; if the payoff of the selected SU is higher than its own payoff, the SU imitates the strategy of the selected SU at the next iteration with a probability proportional to the payoff difference, coefficiented by the imitation factor $\sigma$.\footnote{One way of setting $\sigma$ is to set $\sigma=1/(\omega-\alpha)$, where $\omega$ and $\alpha$ are two exogenous parameters such that $U_j \in [\alpha,\omega], \forall j\in{\cal C}$.}

We first study the dynamic induced by PISAP by setting $\epsilon_U=0$. It is shown in~\cite{Sand10} that in the asymptotic case, the proportional imitation rule in Algorithm~\ref{algo:pir} generates a population dynamic described by the following set of differential equations:
\begin{eqnarray}
\dot x_i(t) = \sigma x_i(t) [\pi_i(t)-\overline{\pi}(t)] \quad i\in{\cal C},
\label{eq:pir_dynamic}
\end{eqnarray}
where $\pi_i$ denotes the expected payoff of the SUs on channel $i$, $\overline{\pi}\triangleq \sum_{i\in{\cal C}}x_i\pi_i$ denotes the expected payoff of all SUs in the network. Injecting $\pi_i=\mu_i/(x_iN)$ into the differential equations, \eqref{eq:pir_dynamic} becomes:
\begin{equation*}
\frac{\dot x_i(t)}{\sigma} = \frac{\mu_i}{N}-x_i(t)\sum_{l\in {\cal C}}\frac{\mu_l}{N}.
\end{equation*}
This equation can be easily solved as:
\begin{equation}
x_i(t) = K_ie^{-\left(\sum_{l\in {\cal C}}\frac{\mu_l}{N}\right)\sigma t}+\frac{\mu_i}{\sum_{l\in {\cal C}} \mu_l},
\label{eq:pir_x}
\end{equation}
where the constant
$K_i = x_i(0)-\frac{\mu_i}{\sum_{l\in {\cal C}} \mu_l}.$

As the first result of this section, the following theorem states the convergence of the dynamic to the NE of the spectrum access game $G$.

\begin{theorem}
\label{th:simple}
The imitation dynamic induced by PISAP converges exponentially to an evolutionary equilibrium which is also the NE of $G$.
\end{theorem}

\begin{proof}
The theorem follows straightforwardly from \eqref{eq:pir_x} and Theorem~\ref{th:ne_congestion_game}.
\end{proof}

As an illustrative example, Figure~\ref{fig:phasePlane} (obtained with \cite{Polkingpplane}) shows the convergence of the imitation dynamic of PISAP to the NE of $G$ for a cognitive network of $2$ channels and $50$ SUs.

We then study the convergence of PISAP in the general case with $\epsilon_U>0$. Specifically, we define the \textit{imitation-stable} equilibrium as a state where no further imitations can be conducted based on the imitation policy~\cite{Ackermann09}. The following theorem analyzes the convergence of PISAP with respect to this concept.

\begin{theorem}
\label{th:convergence_simple}
PISAP converges to an imitation-stable equilibrium in expected $O(\frac{N^2}{\mu_{min}\sigma\epsilon_U})$ iterations where $\mu_{min}\triangleq  \min_{i \in {\cal C}} \mu_i$. The converged equilibrium is an $\epsilon$-NE of $G$ with $\epsilon=2\epsilon_U$.
\end{theorem}

\begin{proof}
We provide the sketch of the proof here. The detailed proof is provided in the Appendix.

We first prove the convergence of PISAP to an imitation-stable equilibrium. Define $i^{max}\triangleq\argmax_{i\in{\cal C}} \pi_i(t)$ and $i^{min}\triangleq\argmin_{i\in{\cal C}} \pi_i(t)$, we show that for each iteration $t$, if $\pi_{i^{max}}(t)-\pi_{i^{min}}(t)>\epsilon_U$, then at least one of the following holds
\begin{eqnarray*}
\begin{cases}
\pi_{i^{min}}(t+1)-\pi_{i^{min}}(t) \sim O(\frac{\mu_{min}\sigma\epsilon_U}{N^2}) \\
\pi_{i^{max}}(t)-\pi_{i^{max}}(t+1) \sim O(\frac{\mu_{max}\sigma\epsilon_U}{N^2})
\end{cases}.
\end{eqnarray*}
I.e., if the difference between the highest expected individual payoff $\pi_{i^{max}}(t)$ and the worst one $\pi_{i^{min}}(t)$ is larger than $\epsilon_U$, we can at least increase $\pi_{i^{min}}(t)$ by $O(\frac{\mu_{min}\sigma\epsilon_U}{N^2})$ or decrease $\pi_{i^{max}}(t)$ by $O(\frac{\mu_{max}\sigma\epsilon_U}{N^2})$. It follows that after a finite number of iterations (the exact form is deduced in the detailed proof), PISAP converges to a state where $\pi_{i^{max}}(t)-\pi_{i^{min}}(t)\le\epsilon_U$, which is imitation-stable. We then show by contradiction that the converged imitation-stable equilibrium is an $\epsilon$-Nash of $G$ with $\epsilon=2\epsilon_U$.
\end{proof}

Note that the convergence delay $O(\frac{N^2}{\mu_{min}\sigma\epsilon_U})$ derived in Theorem~\ref{th:convergence_simple} consists of the upper bound and through the simulations we conduct, we observe that the convergence is achieved in a much shorter delay.

\subsection{Spectrum Access Policy Based on Double Imitation}

In this subsection, we turn to a more advanced imitation rule, the double imitation rule~\cite{Schlag99} and propose the DI-based spectrum access policy, termed as DISAP. Under DISAP, each SU randomly samples two SUs and imitates them with a certain probability determined by the utility difference. The spectrum access policy based on the double imitation is detailed in Algorithm~\ref{algo:apir}, in which each SUs randomly samples two other SUs $j_1$ and $j_2$ (without loss of generality, assume that $j_1$ and $j_2$ operate on channel $i_1$ and $i_2$ respectively, with corresponding utilities $U_{j_1}\le U_{j_2}$) and updates the probabilities of switching to channels $i_1$ and $i_2$, denoted as $p_{j_1}$ and $p_{j_2}$ respectively.

The double imitation rule generates an aggregate monotone dynamic~\cite{Schlag99,Sam92}, which is defined as follows:
\begin{equation}
\dot x_i = \frac{x_i}{\omega-\alpha}\left[1+\frac{\omega-\overline{\pi}}{\omega-\alpha}\right](\pi_i-\overline{\pi}) \quad \forall i \in {\cal C}
\label{eq:aggrDyn}
\end{equation}


Injecting $\pi_i=\mu_i/(x_iN)$ into the differential equations, we have:
\begin{equation*}
\dot x_i = \frac{\sigma\overline{\pi}}{\omega-\alpha}\left(1+\frac{\omega-\overline{\pi}}{\omega-\alpha}\right)-\frac{\sigma\overline{\pi}}{\omega-\alpha}\left(1+\frac{\omega-\overline{\pi}}{\omega-\alpha}\right)x_i,
\end{equation*}
whose solution is
\begin{equation}
x_i(t) = Ke^{-\frac{\sigma\overline{\pi}}{\omega-\alpha}\left(1+\frac{\omega-\overline{\pi}}{\omega-\alpha}\right)t}+\frac{\mu_i}{\sum_{l\in {\cal C}} \mu_l},
\label{eq:di_x}
\end{equation}
where $\overline{\pi}=\sum_{l\in {\cal C}} \mu_l/N$ and
$K = x_i(0)-\frac{\mu_i}{\sum_{l\in {\cal C}} \mu_l}$. In the studied scenario, $\alpha$ and $\omega$ are the lower and upper bound of the SUs' utility, which are $0$ and $1$, respectively.

The following theorem stating the major result in this subsection follows immediately.

\begin{theorem}
DISAP converges exponentially to the NE of the spectrum access game $G$.
\end{theorem}

Compared with the proportional imitation rule, which produces the replicator dynamic (Eq.~\eqref{eq:pir_dynamic}), the adjusted proportional imitation rule induces the aggregate monotone
dynamic (Eq.~\eqref{eq:aggrDyn}) that converges to the NE at a higher rate. In Fig.~\ref{fig:phasePlane} is shown the phase plane for replicator dynamics and aggregate monotone dynamics. As proven in Theorem~\ref{th:ne_congestion_game}
for large $N$ there exists only one attractor (NE), which is the crossing point of the two nulclines (dashed lines).


We then study the convergence to an imitation-stable equilibrium of DISAP in the general case with $\epsilon_U>0$ in the following theorem.

\begin{theorem}
\label{th:convergence_double}
DISAP converges to an imitation-stable equilibrium in expected $O(\frac{N^2}{\mu_{min}\sigma\epsilon_U})$ iterations where $\quad$ $\mu_{min}\triangleq  \min_{i \in {\cal C}} \mu_i$. The converged equilibrium is an $\epsilon$-NE of $G$ with $\epsilon=2\epsilon_U$.
\end{theorem}

\begin{proof}
The proof follows the same analysis as that of Theorem~\ref{th:convergence_simple}. The detail is provided in the Appendix.
\end{proof}

\subsection{Discussion}

As desirable properties, the proposed imitation-based spectrum access policies (both PISAP and DISAP) are stateless, incentive-compatible for selfish autonomous SUs and requires no central computational unit. The spectrum assignment is achieved by local interactions among autonomous SUs and the $\epsilon$-optimum of the system is achieved when the algorithm converges, which is achieved in polynomial time. The autonomous behavior and decentralized implementation make the proposed policies especially suitable for large scale cognitive radio networks. The imitation factor $\sigma$ controls the tradeoff between the convergence speed and the channel switching frequency in that larger $\sigma$ represents more aggressiveness in imitation and thus leads to fast convergence, at the price of more frequent channel switching for the SUs which may consist of significant cost for today's wireless devices in terms of delay, packet loss and protocol overhead. The imitation threshold $\epsilon_U$, on the other hand, can be tuned to balance between the convergence speed and the optimality of the converged equilibrium. 

\section{Imitation on the same channel}
\label{sec:NewScenario}



Up to now, we have studied the imitation-based channel access policy where a SU can imitate any other SU whatever the channel the latter stays in. This approach implicitly assumes that a  SU can interact with SUs on different channels, which may not be realistic in some cases or pose additional system overhead (e.g., sensing a different channel). In this section, we focus on a more practical scenario, where a SU only imitates the SUs on the same channel and the imitation is based on the payoff difference of the precedent iteration. In the considered scenario, a SU only needs to locally interact with the SUs on the same channel (e.g., exchange payoff of the precedent iteration, which can be piggybacked with the data packets transmitted on the channel).


In the sequel analysis, we first study the induced imitation dynamic and the convergence of the proposed spectrum access policies PISAP and DISAP subject to channel constraint on imitation.

\subsection{Imitation Dynamic and Convergence}
\label{subsec:dynamicsScenario2}


In this subsection, we first derive in Theorem~\ref{th:finitePop} the dynamic for a generic imitation rule $F$ with large population. We then derive in Lemma~\ref{lemma:pop2}, Theorem~\ref{th:pop2} and Theorem~\ref{th:convergence_pi_2} the dynamic of the proposed proportional imitation policy PISAP and its convergence under the channel constraint. The counterpart analysis for the double imitation policy DISAP is explored in Lemma~\ref{lemma:pop3}, Theorem~\ref{th:pop3} and Theorem~\ref{th:convergence_di_2}.




We start by introducing the notations used in our analysis. At an iteration, we label all SUs performing strategy $i$ (channel $i$ in our case) as SUs of type $i$ and we refer to the SUs on $s_j$ as neighbors of SU~$j$. We denote $n_{i}^l(t)$ the number of SUs on channel $i$ at iteration $t$ and operating on channel $l$ at $t-1$. It holds that $\sum_{l\in{\cal C}} n_i^l(t) = n_i(t)$ and $\sum_{i\in{\cal C}} n_i^l(t) = n_l(t-1)$. For a given state $s(t)\triangleq\{s_j(t), j\in{\cal C}\}$ at iteration $t$ and a finite population of size $N$, we denote
$p_i(t)\triangleq n_i(t)/N$ the proportion of SUs of type $i$ and $p_{i}^l(t)\triangleq n_i^l(t)/N$ the proportion of SUs migrating from channel $l$ to $i$. We use $x$ instead of $p$ to denote these proportions in asymptotic case. It holds that $p\rightarrow x$ when $N\rightarrow +\infty$.


In our study, a generic imitation rule under the channel constraint is termed as $F$. In the case of the proportional imitation rule (c.f. PISAP), $F$ is characterized by the probability set
$\{F_{j,k}^i\}$ where $F_{j,k}^i$ denotes the probability that a SU choosing strategy $j$ at the precedent iteration imitates another SU choosing strategy $k$ at the precedent iteration and
then switches to channel $i$ at next iteration after imitation. Instead, by applying the double imitation rule (c.f. DISAP), we can characterize $F$ by the probability set $\{F_{j,\{k,l\} }^i\}$
where $F_{j,\{k,l\} }^i$ denotes the probability that a SU choosing strategy $j$ at the precedent iteration imitates two neighbors choosing respectively strategy $k$ and strategy $l$ at the
precedent iteration and then switches to channel $i$
at next iteration after imitation. In both cases the only way to switch to a channel $i$ is to imitate a SU that was on channel $i$. That means $F_{j,k}^i=0$, $\forall k\ne i$ (PISAP)
and $F_{j,\{k,l\} }^i=0$, $\forall k,l\ne i$ (DISAP). \\
At the initialization phase (iteration $0$ and $1$), each SU randomly chooses its strategy. After that, the system state at iteration $t+1$, denoted
as $\mathbf{p(t+1)}$ ($\mathbf{x(t+1)}$ in the asymptotic case), depends on the states at iteration $t$ and $t-1$.


\begin{theorem}
For any imitation rule $F$, if the imitation among SUs of the same type occurs randomly and independently, then $\forall \delta>0$, $\epsilon>0$ and any initial state $\{\widetilde{x}_i(0)\}$, $\{\widetilde{x}_i(1)\}$, there exists $N_{0}\in\mathbb{N}$
such that if $N>N_{0}$, $\forall i\in{\cal C}$, the event $|p_i(t)-x_i(t)|>\delta$ occurs with probability less than $\epsilon$, where
$p_i(0)=x_i(0)=\widetilde{x}_i(0)$, $p_i(1)=x_i(1)=\widetilde{x}_i(1)$. In the case of proportional imitation policy it holds that
\begin{equation*}
x_i(t+1)=\sum_{j,l,k \in \mathcal{C}}  \frac{x_j^l(t)x_j^k(t)}{x_j(t)}F_{l,k}^{i} \quad \forall i\in \mathcal{C}
\end{equation*}
Differently, the double imitation policy yields:
\begin{equation*}
x_i(t+1)=\sum_{j,l,k,z \in \mathcal{C}}  \frac{x_j^l(t)x_j^k(t)x_j^z(t)}{[x_j(t)]^2}F_{l,\{k,z\} }^{i} \quad \forall i\in \mathcal{C}
\end{equation*}

\label{th:finitePop}
\end{theorem}

\begin{proof}
The proof consists of first showing the theorem holds for iteration $t=2$ and then proving the case $t\ge 3$ by induction. The detail is in Appendix.
\end{proof}

Theorem~\ref{th:finitePop} is an important result on the short run adjustments of large populations under any generic imitation rule $F$: the probability that the behavior
of a large population differs from the one of an infinite population is arbitrarily small when $N$ is sufficiently large. In what follows, we study the convergence of PISAP and DISAP under the channel constraint.

\noindent \textbf{(1) Spectrum access policy PISAP under channel constraint}

We now focus on PISAP under channel constraint and derive the induced imitation dynamic by setting $\epsilon_U=0$ in the following analysis.

\begin{lemma}
On the proportional imitation policy PISAP under channel constraint, it holds that
\begin{equation}
x_i^j(t+1) = \sum_{l,k \in \mathcal{C}}  \frac{x_j^l(t)x_j^k(t)}{x_j(t)}F_{l,k}^{i} \quad \forall i,j\in \mathcal{C}.
\label{xijlemma}
\end{equation}
\label{lemma:pop2}
\end{lemma}
\begin{proof}
The proof is straightforward from the analysis in the proof of Theorem~\ref{th:finitePop}.
\end{proof}



\begin{theorem}
The proportional imitation policy PISAP under channel constraint generates the following dynamic in the asymptotic case
\begin{equation}
\label{eq:pop2}
x_i(t+1)=x_i(t-1)+\sigma\pi_i(t-1)x_i(t-1)-\sigma\sum_{j,l\in\mathcal{C}}\pi_l(t-1)\frac{x_j^i(t)x_j^l(t)}{x_j(t)}
\end{equation}
where $\pi_i(t)$ denotes the expected payoff of an individual SU on channel $i$ at iteration $t$.
\label{th:pop2}
\end{theorem}

\begin{proof}
It can be shown that the proportional imitation policy PISAP under channel constraint is both {\it imitating}\footnote{A behavior rule is imitating if the switching actions of the rule occur by imitating the sampled individual.} and {\it improving}\footnote{A behavior rule is improving if and only the expected payoff of an individual increases after imitation.}. Apply the analysis in \cite{Schlag96} (Eq.~(10)) to our case, we can characterize $\{F_{l,k}^i\}$ for the proportional imitation rule under channel constraint as:
\begin{equation*}
F_{l,k}^i=
\begin{cases}
0 & k\neq i \\
F_{i,l}^l+\sigma[\pi_i(t-1)-\pi_l(t-1)] & k=i
\end{cases}.
\end{equation*}
In other words, the only possibility to switch to channel $i$ is to imitate a SU that was on channel $i$; the switching probability is proportional to the payoff difference. Noticing that $\sum_l x_j^l(t)F_{i,l}^l/x_j(t)=1$, \eqref{xijlemma} can be written as follows:
\begin{eqnarray*}
x_i^j(t+1) &=& \sum_{l \in \mathcal{C}}  \frac{x_j^l(t)x_j^i(t)}{x_j(t)}F_{l,i}^i = \sum_{l \in \mathcal{C}}  \frac{x_j^l(t)x_j^i(t)}{x_j(t)}F_{i,l}^l+\sum_{l \in \mathcal{C}}  \frac{x_j^l(t)x_j^i(t)}{x_j(t)}\sigma[\pi_i(t-1)-\pi_l(t-1)] \\
           &=& x_j^i(t) + \sum_{l \in \mathcal{C}}  \frac{x_j^l(t)x_j^i(t)}{x_j(t)}\sigma[\pi_i(t-1)-\pi_l(t-1)].
\end{eqnarray*}

Injecting $\sum_j x_j^i(t)=x_i(t-1)$, $\sum_l x_j^l(t)/x_j(t)=1$ and $x_i(t+1)=\sum_j x_i^j(t+1)$ into the above formula, we can obtain \eqref{eq:pop2}, which concludes the proof.
\end{proof}

We observe via extensive numerical experiments that \eqref{eq:pop2} always converges to an evolutionary equilibrium. To get more in-depth insight on the dynamic \eqref{eq:pop2}, we notice that under the following approximation:
\begin{equation}
\sum_{l \in \mathcal{C}} \pi_l(t-1)\frac{x_j^l(t)}{x_j(t)} \approx \bar{\pi}(t-1),
\label{eq:approximation}
\end{equation}
where $\bar{\pi}(t-1)$ is the average individual payoff for the whole system at iteration $t-1$, noticing $\sum_j x_j^i(t)=x_i(t-1)$, \eqref{eq:pop2} can be written as:
\begin{equation}
x_i(t+1)=x_i(t-1)+\sigma x_i(t-1)[\pi_i(t-1)-\bar{\pi}(t-1)].
\label{eq:dynamic1}
\end{equation}
Note that the approximation~\eqref{eq:approximation} states that in any channel $j$ at iteration $t$, the proportions of SUs coming from any channel $l$ are representative of the whole population.

Under the approximation~\eqref{eq:approximation}, given the initial state $\{x_i(0)\}$, $\{x_i(1)\}$, we can decompose~\eqref{eq:dynamic1} into the following two independent discrete-time replicator dynamics:
\begin{equation}
\begin{cases}
x_i(u)=x_i(u-1)+\sigma x_i(u-1)[\pi_i(u-1)-\bar{\pi}(u-1)] \\
x_i(v)=x_i(v-1)+\sigma x_i(v-1)[\pi_i(v-1)-\bar{\pi}(v-1)]
\end{cases}
\label{eq:system_pi}
\end{equation}
where $u=2t$, $v=2t+1$. The two equations in \eqref{eq:system_pi} illustrate the underlying system dynamic hinged behind the proportional imitation policy under channel constraint under the approximation~\eqref{eq:approximation}: it can be decomposed into two independent delayed replicator dynamics that alternatively occur at the odd and even iterations, respectively. The following theorem establishes the convergence of \eqref{eq:system_pi} to a unique fixed point which is also the NE of the spectrum access game $G$.

\begin{theorem}
Starting from any initial point, the system described by \eqref{eq:system_pi} converges to a unique fixed point which is also the NE of the spectrum access game $G$.
\label{th:convergence_pi_2}
\end{theorem}

\begin{proof}
The proof, of which the detail is provided in the Appendix, consists of showing that the mapping described by \eqref{eq:system_pi} is a contraction mapping.
\end{proof}

As an illustrative example, Figure~\ref{fig:approximation} shows that the double replicator dynamic provides an accurate approximation of the system dynamic induced by PISAP under channel constraint.

%

Furthermore, performing the same analysis as that of Theorem~\ref{th:convergence_simple}, we can establish the same convergence property on the imitation algorithm under channel constraint under the approximation~\eqref{eq:approximation} for the general case with $\epsilon_U\ge 0$.

\noindent \textbf{(2) Spectrum access policy DISAP under channel constraint}

We then focus on DISAP under channel constraint and derive the induced imitation dynamic.

\begin{lemma}
On the double imitation policy DISAP under channel constraint, it holds that
\begin{equation}
x_i^j(t+1) = \sum_{l,k,z \in \mathcal{C}}  \frac{x_j^l(t)x_j^k(t)x_j^z(t)}{[x_j(t)]^2}F_{l,\{k,z\} }^{i} \quad \forall i,j\in \mathcal{C}.
\label{xijlemma2}
\end{equation}
\label{lemma:pop3}
\end{lemma}
\begin{proof}
The proof is straightforward from the analysis in the proof of Theorem~\ref{th:finitePop}.
\end{proof}



\begin{theorem}
The double imitation policy DISAP under channel constraint generates the following dynamic in the asymptotic case
\begin{multline}
\label{eq:pop3}
x_i(t+1) = x_i(t-1)+2x_i(t-1)\pi_i(t-1)+\sum_jx_j^i(t)\left[ \sum_k\frac{x_j^k(t)}{x_j(t)}\right]^2-2x_j^i(t)\sum_k\frac{x_j^k(t)}{x_j(t)}\pi_k(t-1)\\
-x_j^i(t)\pi_i(t-1)\sum_k\frac{x_j^k(t)}{x_j(t)}\pi_k(t-1)
\end{multline}
where $\pi_i(t)$ denotes the expected payoff of an individual SU on channel $i$ at iteration $t$.
\label{th:pop3}
\end{theorem}

\begin{proof}
If the rule $F$ is {\it unbiased}\footnote{A behavioural rule is {\it unbiased} if it does not depend on the labelling of actions.}, it follows from \cite{Schlag99} that $F$ is also {\it improving}
and {\it globally efficient}\footnote{A behavioural rule is globally efficient if, for any Multi-armed bandit, all individuals use a best action in the long run, provided that initially
each action is present.}. We can than characterize $\{F_{l,k}^i\}$ as in \cite{Schlag99} (Theorem~1):

\begin{equation*}
F_{l,\{k,i\}}^i=F_{i,\{l,k\}}^k+F_{i,\{k,l\}}^l-F_{k,\{l,i\}}^i-\frac{1}{2}\sigma(\pi_k)(\pi_l-\pi_i)-\frac{1}{2}\sigma(\pi_l)(\pi_k-\pi_i)
\end{equation*}

In other words, the only possibility to switch to channel $i$ is to imitate a SU that was on channel $i$.
Setting $\omega=1$ and $\alpha=0$, so that $\sigma(y)=2-y$, and noticing that $\sum_{l,k} x_j^l(t)x_j^k(t)F_{k,\{ l,i\} }^i/x_j(t)^2=1$, \eqref{xijlemma} can be written as follows:
\begin{equation*}
x_i^j(t+1)=x_j^i(t)+x_j^i(t)\left( 2-\sum_k\frac{x_j^k(t)\pi_k(t-1)}{x_j}\right)\left( \pi_i(t-1)-\sum_k\frac{x_j^k\pi_k(t-1)}{x_j}\right)
\end{equation*}

Injecting $\sum_j x_j^i(t)=x_i(t-1)$, $\sum_l x_j^l(t)/x_j(t)=1$ and $x_i(t+1)=\sum_j x_i^j(t+1)$ into the above formula, we can obtain \eqref{eq:pop3}, which concludes the proof.
\end{proof}

We observe via extensive numerical experiments that \eqref{eq:pop3} always converges to an evolutionary equilibrium and, as shown in Fig,~\ref{fig:DIvsPIR}, also
features a smoother and faster convergence trend with respect to the proportional imitation dynamic (Eq.~\eqref{eq:pop2}).\\
By performing the same approximation as in \eqref{eq:approximation},
\eqref{eq:pop3} can be written as:
\begin{equation}
x_i(t+1)=x_i(t-1)+x_i(t-1)(2-\bar{\pi}(t-1))(\pi_i(t-1)-\bar{\pi}(t-1)).
\label{eq:dynamic}
\end{equation}

Under the approximation~\eqref{eq:approximation}, given the initial state $\{x_i(0)\}$, $\{x_i(1)\}$, we can decompose~\eqref{eq:dynamic} into the following two independent discrete-time aggregate monotone dynamics:
\begin{equation}
\begin{cases}
x_i(u)=x_i(u-1)+ x_i(u-1)[2-\bar{\pi}(u-1)]\cdot[\pi_i(u-1)-\bar{\pi}(u-1)]\\
x_i(v)=x_i(v-1)+ x_i(v-1)[2-\bar{\pi}(v-1)]\cdot[\pi_i(v-1)-\bar{\pi}(v-1)]
\end{cases}
\label{eq:system_di}
\end{equation}
where $u=2t$, $v=2t+1$. The above two equations illustrate the underlying system dynamic hinged behind the double imitation policy under channel constraint under the approximation~\eqref{eq:approximation}: it can be decomposed into two independent delayed aggregate monotone dynamics that alternatively occur at the odd and even iterations, respectively. The following theorem establishes the convergence of \eqref{eq:system_di} to a unique fixed point which is also the NE of the spectrum access game $G$. The proof follows exactly the same analysis as that of Theorem~\ref{th:convergence_pi_2}.

\begin{theorem}
Starting from any initial point, the system described by \eqref{eq:system_di} converges to a unique fixed point which is also the NE of the spectrum access game $G$.
\label{th:convergence_di_2}
\end{theorem}

As an illustrative example, Figure~\ref{fig:DIvsAPPROX} shows that the double aggregate dynamic provides an accurate approximation of the system dynamic induced by the
double imitation under channel constraint.

\subsection{Imitation-Based Channel Access Policy under Channel Constraint}
\label{subsec:algorithm}

In this subsection, based on the theoretic results derived previously, we develop a fully distributed channel access policy for the general case with finite population based on the imitation rule among SUs on the same channel (i.e. neighbors). The proposed policy, detailed in Algorithm~\ref{algo:generic_imitation}, is suitable both for proportional and double imitation. Run at each SU $j$ and at each iteration, it consists of:
\begin{compactitem}
\item sampling randomly one (proportional imitation) or two (double imitation) neighbors;
\item comparing the payoff achieved at the previous iteration $t-1$ with that of the neighbor(s) selected for imitation;
\item performing channel migration with the probability dictated by the applied imitation rule.
\end{compactitem}
Algorithm~\ref{algo:generic_imitation} is evaluated by extensive simulations in next section.

 \section{Performance Evaluation}
 \label{sec:simu}
 In this section we conduct extensive simulations to evaluate the performance of the proposed imitation-based channel access policy (PISAP and DISAP) in both scenarios with and without channel constraints and demonstrate some intrinsic properties of the policy which are not explicitly addressed in the analytical part of the paper.

 \subsection{Simulation Settings}

 We simulate a cognitive radio network of $N=50$ SUs and $C=3$ channels, on which PUs has different activity rates on different channels, leading to different channel availability
probabilities characterized by $\mu=[0.3, 0.5, 0.8]$. We assume the iteration time to be long enough so that the SUs, regardless of the occupied channel, can evaluate their payoff
without errors.


 \subsection{System Dynamics}

We first study the system dynamic induced by PISAP and DISAP in the scenarios with and without channel constraint.

As illustrated by Fig.~\ref{fig:oursVSreplicators} and Fig.~\ref{fig:oursVSaggregate},
both \eqref{eq:pop2} and \eqref{eq:pop3}, which reflect PISAP and DISAP behaviour in the asymptotic case with channel constraint, produce trajectories that converge in a faster but less smooth manner if compared with
their respective unconstrained dynamics. This can be interpreted by the overlap of two replicator/aggregate monotone dynamics at odd and even instants, as
explained in section~\ref{sec:NewScenario}.\\
In Fig.~\ref{fig:DIvsPIR} the trends of \eqref{eq:pop2} and \eqref{eq:pop3} are compared. We observe that, in the asymptotic case, DISAP outperforms PISAP as it is characterized by less
pronounced wavelets and a faster convergence. However, all the displayed dynamics correctly converge
to an evolutionary equilibrium. It is easy to check that the converged equilibrium is also the NE of G and the system optimum, which confirms our theoretic analysis.

\subsection{Convergence to Imitation-stable Equilibrium}
We now study the convergence of PISAP and DISAP with and without channel constraint for a finite number of users ($N=50$). Fig.~\ref{fig:pirSuxCHwoCC} and Fig.~\ref{fig:diSuxCHwoCC} show the number
of SUs per channel during the convergence phase for one realization of our algorithms without channel constraint.
We notice that in both cases convergence is rapidly achieved after few iterations, and that the channels with higher availabilities are
chosen by more individuals. This can be easily verified (Fig.~\ref{fig:pirSuxCHwoCC} for PISAP and Fig.~\ref{fig:diSuxCHwoCC} for DISAP) by observing that after
convergence the major part of population settles permanently in channel 3, i.e. the channel that less frequently hosts PU's transmissions.

Fig.~\ref{fig:PIR_convergencePH} and Fig.~ \ref{fig:DI_convergencePH} show that, starting from the same initial conditions, DISAP reaches an imitation-stable equilibrium more rapidly than PISAP,
at the price of a higher algorithmic complexity and a substantial increase in information exchanges amongst the SUs due to double imitation.
Note that the small deviation of the trajectories at some iterations in the figures from the converged curve is due to the
probabilistic nature of the users' strategy and has only very limited impact on the system as a whole.

Fig.~\ref{fig:PIR_SIneighbors_NOlearning} and Fig.~\ref{fig:DI_SIneighbors_NOlearning} show instead a realization of our algorithms with channel constraint.
We notice that an imitation-stable equilibrium is achieved progressively following the dynamics characterized by~\eqref{eq:pop2} and \eqref{eq:pop3}.
The equilibrium is furthermore very close to the system optimum: we can in fact check that, according to Theorem~\ref{th:ne_congestion_game}, the proportion of SUs choosing channel $1$, $2$ and $3$ at the system optimum is $0.1875$, $0.3125$ and $0.5$ respectively (see also Fig.~\ref{fig:oursVSreplicators} and Fig.~\ref{fig:oursVSaggregate});
in the simulation results we observe that there are $9$, $16$ and $25$ SUs settling on channel $1$, $2$ and 3 respectively.

\subsection{System Fairness}


We now turn to the analysis of the fairness of the proposed spectrum access policies. To this end, we adopt the Jain's fairness index~\cite{jains84}, which varies in $[0, 1]$ and reaches its maximum when
the resources are equally shared amongst users. Fig.~\ref{fig:jain}, whose curves represent an average over $10^3$ independent realizations of our algorithms, shows that our system turns out to be very fair
even from the early iterations.\\
From the same figures one can further infer that indeed DISAP converges more rapidly than PISAP: if for instance we fix on the y-axis the fairness value $0.982$, the latter
is reached by DISAP at the iteration $t=100$, and by PISAP at the iteration $t=200$.

\section{Related Work}
\label{sec:related_work}

The spectrum access problem in the considered gene\-ric model (with single SU) is closely related to the classic Multi-Armed Bandit (MAB) problem~\cite{Mahajan07}. In this case, a SU should strike a balance between exploring the environment to find profitable channels (i.e., learn the channel availability probabilities) and exploiting the best one using current knowledge. In this line of research, Gittins developed an index policy in~\cite{Gittins89} that consists of selecting the arm with the highest index termed as Gittins index. This policy is shown to be optimal in the most general case. Lai and Robbins~\cite{Lai85} and then Agrawal~\cite{Agrawal95} studied the MAB problem by proposing policies based on the {\it upper confidence bounds} with logarithmic regret. Compared with the classic MAB problem, one major specialty of the spectrum access in cognitive radio networks lies in the fact of multiple SUs that can cause collisions if they simultaneously access the same channel. Some recent work has investigated this issue, among which Anandkumar \textit{et al.} proposed two algorithms with logarithmic regret, where the number of SUs is known~\cite{Anand09} and unknown and estimated by each SU~\cite{Anand10}, Liu and Zhao developed a time-division fare share (TDFS) algorithm with convergence and logarithmic regret~\cite{Liu09}. 

As the spectrum access problem in cognitive radio is essentially a resource allocation problem, another important thrust consists of applying game theory to model the competition and cooperation (coordination) among SUs and the interaction between SUs and PUs. Along this line of research, a game theoretic framework was developed in~\cite{Nie06} to analyze the behavior of selfish users in cognitive radio networks, resulting in distributed adaptive channel allocation algorithms based on the technique of no-regret learning. In~\cite{Neel04}, the convergence of different types of games in cognitive radio systems was studied (i.e., coordinated behavior, best response, best response for discounted repeated games, S-modular games and potential games). A no-regret learning algorithm was proposed in~\cite{Han07} to address the channel allocation problem in cognitive networks. The algorithm can reach a correlated equilibrium which, in many case, is more efficient than the classic Nash equilibrium of the game. Maskery \textit{et al.} considered the dynamic spectrum access among cognitive radios from an adaptive, game theoretic learning perspective and proposed decentralized dynamic spectrum access protocol~\cite{Maskery09}. Besides, due to the perceived fairness and allocation efficiency, auction techniques have also attracted considerable research attention and resulted in a number of auction-based spectrum allocation mechanisms (cf.~\cite{infocom10} and references therein).

Due to the success of applying evolutionary game theory in the study of biological and economic problems, a handful of recent studies have applied evolutionary game theory as a tool to study resource allocation problems arisen from wired and wireless networks, among which Shakkottai~\textit{et al.} addressed the problem of non-cooperative multi-homing of users to access points in IEEE 802.11 WLANs by modeling it as a population game and studied the equilibrium properties of the game~\cite{Shakkottai07}; Niyato \textit{et al.} studied the dynamics of network selection in a heterogeneous wireless network using the theory of evolutionary game and the replicator dynamic and proposed two network selection algorithm to reach the evolutionary equilibrium~\cite{Niyato09}; Ackermann \textit{et al.} investigated the concurrent imitation dynamics in the context of symmetric congestion games by focusing on the convergence properties~\cite{Ackermann09}; Niyato \textit{et al.} studied the multiple-seller and multiple-buyer spectrum trading game in cognitive radio networks using the replicator dynamic and provided a theoretic analysis for the two-seller two-group-buyer case~~\cite{Niyato09b}. Coucheney~\textit{et al.} studied the user-network association problem in wireless networks with multi-technology and proposed an algorithm to achieve the fair and efficient solution~\cite{Cou09}.

\section{Conclusion and Further Work}
\label{sec:conclusion}

In this paper, we address the spectrum access problem in cognitive radio networks by applying evolutionary game theory and develop an imitation-based
spectrum access policy. We investigate two imitation scenarios where a SU can imitate any other SUs and where it can only imitate the other SUs operating
on the same channel. A systematic theoretical analysis is presented for both scenarios on the induced imitation
dynamics and the convergence properties of the proposed policy to an imitation-stable equilibrium, which is also the $\epsilon$-optimum of the system.
As an important direction of the future work, we plan to investigate the imitation-based channel access problem in the more generic multi-hop scenario where SUs can imitate their neighbors and derive the relevant channel access/assignment policies there.



%
%
%

\bibliographystyle{unsrt}
\bibliography{imitation}

\appendix

\begin{proof}[\textbf{Proof of Theorem~\ref{th:convergence_simple}}]

We first prove the convergence of PISAP to an imitation-stable equilibrium.

Let $i^{max}\triangleq\argmax_{i\in{\cal C}} \pi_i(t)$, $i^{min}\triangleq\argmin_{i\in{\cal C}} \pi_i(t)$, we show that for any iteration $t$, if $\pi_{i^{max}}(t)-\pi_{i^{min}}(t)>\epsilon_U$, then at least one of the following holds
\begin{eqnarray*}
\begin{cases}
\pi_{i^{min}}(t+1)-\pi_{i^{min}}(t) \sim O(\frac{\mu_{i^{min}}\sigma\epsilon_U}{N^2}) \\
\pi_{i^{max}}(t)-\pi_{i^{max}}(t+1) \sim O(\frac{\mu_{i^{max}}\sigma\epsilon_U}{N^2})
\end{cases}.
\end{eqnarray*}
I.e., if the difference between the highest expected individual payoff $\pi_{i^{max}}(t)$ and the worst one $\pi_{i^{min}}(t)$ is larger than $\epsilon_U$, we can at least increase $\pi_{i^{min}}(t)$ by $O(\frac{\mu_{i^{min}}\sigma\epsilon_U}{N^2})$ or decrease $\pi_{i^{max}}(t)$ by $O(\frac{\mu_{i^{max}}\sigma\epsilon_U}{N^2})$.


Define the $\epsilon_U$-worst channel set ${\cal C}_{\epsilon_U}$ as the channel set such that a channel $i\in{\cal C}_{\epsilon_U}$ iff $\pi_i-\pi_{i^{min}}\le \epsilon_U$ and $\overline{\pi}-\pi_i> \epsilon_U/2$. In the same way, define the $\epsilon_U$-best channel set ${\cal C}^{\epsilon_U}$ as the channel set such that a channel $i\in{\cal C}^{\epsilon_U}$ iff $\pi_{i^{max}}-\pi_i \le \epsilon_U$ and $\pi_i-\overline{\pi}> \epsilon_U/2$. At any iteration $t$, if $\pi_{i^{max}}(t)-\pi_{i^{min}}(t)>\epsilon_U$, at least one of ${\cal C}_{\epsilon_U}$ and ${\cal C}^{\epsilon_U}$ is not empty. Without loss of generality, assume that ${\cal C}_{\epsilon_U}\ne \emptyset$.

For any channel $i\in{\cal C}_{\epsilon_U}(t)$, let $x^i_l(t+1)$ denote the proportion of SUs migrating from channel $i$ to $l$ after the imitation in the iteration $t$, i.e., the proportion of SUs operating on channel $i$ at iteration $t$ and switching to channel $l$ for iteration $t+1$. Given that the probability of imitating a SU in channel $l$ is $x_l(t)$, $x^i_l(t+1)$ can be computed as
\begin{equation*}
x^i_l(t+1)=
\begin{cases}
x_i(t)x_l(t)\sigma[\pi_l(t)-\pi_i(t)] & l\in{\cal C}-{\cal C}_{\epsilon_U}(t) \\
0 & l\in{\cal C}_{\epsilon_U}(t)
\end{cases}.
\end{equation*}

Denote $\Delta x_i(t)\triangleq x_i(t+1)-x_i(t)$, it follows from the imitation rule that no SU migrates to channel $i^{min}$. Hence, $\Delta x_{i^{min}}(t)<0$ and
\begin{eqnarray*}
-\mathbb{E}[\Delta x_{i^{min}}(t)]&=&\sum_{l\in{\cal C}-{\cal C}_{\epsilon_U}(t)} \hspace{-0.3cm}\mathbb{E}[x^{i^{min}}_l(t+1)]=x_{i^{min}}(t)\sigma\left[\sum_{l\in{\cal C}-{\cal C}_{\epsilon_U}(t)} \hspace{-0.3cm}x_l(t)\pi_l(t)-\sum_{l\in{\cal C}-{\cal C}_{\epsilon_U}(t)} \hspace{-0.3cm}x_l(t)\pi_{i^{min}}(t)\right] \\
&=& x_{i^{min}}(t)\sigma\sum_{l\in{\cal C}-{\cal C}_{\epsilon_U}(t)} x_l(t)[\pi_l(t)-\pi_{i^{min}}(t)].
\end{eqnarray*}

Since ${\cal C}_{\epsilon_U}(t)\ne\emptyset$, there exists at least a channel $l_0$ with $\pi_{l_0}-\pi_{i^{min}}>\epsilon_U$ and $x_{l_0}(t)\ge 1/N$. We have
\begin{equation}
-\mathbb{E}[\Delta x_{i^{min}}(t)] > x_{i^{min}}(t)\sigma x_{l_0}(t)[\pi_{l_0}(t)-\pi_{i^{min}}(t)]> \frac{x_{i^{min}}(t)\sigma\epsilon_U}{N}.
\label{eq:prove_reasonning}
\end{equation}

It then follows that
\begin{eqnarray*}
\pi_{i^{min}}(t+1)-\pi_{i^{min}}(t) &=& \frac{\mu_{i^{min}}}{Nx_{i^{min}}(t+1)}-\frac{\mu_{i^{min}}}{Nx_{i^{min}}(t)} = \frac{1}{N}\left[\frac{\mu_{i^{min}}}{x_{i^{min}}(t)+\Delta x_{i^{min}}(t)}-\frac{\mu_{i^{min}}}{x_{i^{min}}(t)}\right] \\
                                    & >& \frac{\mu_{i^{min}}[-\Delta x_{i^{min}}(t)]}{N[x_{i^{min}}(t)]^2}>\frac{\mu_{i^{min}}\sigma\epsilon_U}{N^2x_{i^{min}}(t)}>\frac{\mu_{i^{min}}\sigma\epsilon_U}{N^2}
\end{eqnarray*}

Given that $\pi_i\in[0,1], \forall i\in{\cal C}$ and let $\mu_{min}\triangleq \min_{i \in {\cal C}} \mu_i$, it holds that after at most $O(\frac{N^2}{\mu_{min}\sigma\epsilon_U})$ iterations, PISAP converges to an imitation-stable equilibrium where $\pi_{i^{max}}-\pi_{i^{min}}\le\epsilon_U$.

We then show by contradiction that the imitation-stable equilibrium is an $\epsilon$-Nash of $G$ with $\epsilon=2\epsilon_U$.

Denote $x_i(\infty)$ and $\pi_i(\infty)$ the converged values of $x_i$ and $\pi_i$ in PISAP. Denote $\mathbf{x^*}=\{x_i^*\}$ the NE of $G$. Recall that all payoffs at NE are equal: $\mu_i/Nx_i^*$ is a constant. Assume now that there exists channel $i_0$ such that
$$|\pi_{i_0}(\infty)-\mu_{i_0}/Nx^*_{i_0}|>2\epsilon_U.$$
Without loss of generality, assume that
$$\pi_{i_0}(\infty)-\mu_{i_0}/Nx^*_{i_0}>2\epsilon_U.$$
We have now:
\begin{eqnarray*}
\forall i\in{\cal C}, \; \pi_i(\infty)-\mu_i/Nx_i^*&=&\pi_{i_0}(\infty)-\mu_i/Nx_i^*+\pi_i(\infty)-\pi_{i_0}(\infty) \\
&=&(\pi_{i_0}(\infty)-\mu_{i_0}/Nx_{i_0}^*)+(\pi_i(\infty)-\pi_{i_0}(\infty)) \\
&>&  2\epsilon_U - |\pi_{i^{max}}-\pi_{i^{min}}| \\
&>&2\epsilon_U-\epsilon_U = \epsilon_U.
\end{eqnarray*}

This leads to $x_i(\infty)<x_i^*, \forall i\in{\cal C}$ and $\sum_{i\in{\cal C}} x_i(\infty)=1<\sum_{i\in{\cal C}} x_i^*=1$, which is clearly a contradiction.
\end{proof}

\begin{proof}[\textbf{Proof of Theorem~\ref{th:convergence_double}}]

The proof follows the same analysis as that of Theorem~\ref{th:convergence_simple}. With the same notation in the proof of Theorem~\ref{th:convergence_simple}, we first show that for each iteration $t$, if $\pi_{i^{max}}(t)-\pi_{i^{min}}(t)>\epsilon_U$, then at least one of the following holds
\begin{eqnarray*}
\begin{cases}
\pi_{i^{min}}(t+1)-\pi_{i^{min}}(t) \sim O(\frac{\mu_{i^{min}}\sigma\epsilon_U}{N^2}) \\
\pi_{i^{max}}(t)-\pi_{i^{max}}(t+1) \sim O(\frac{\mu_{i^{max}}\sigma\epsilon_U}{N^2})
\end{cases}.
\end{eqnarray*}

Under the double imitation, $x^i_l(t+1)$ can be computed as
\begin{equation*}
x^i_l(t+1)=
\begin{cases}
x_i(t)x_l(t)p_{l} & l\in{\cal C}-{\cal C}_{\epsilon_U}(t) \\
0 & l\in{\cal C}_{\epsilon_U}(t)
\end{cases}.
\end{equation*}

Injecting $p_l$ into the above formula, after some algebraic operations and following the same reasoning as that in~\eqref{eq:prove_reasonning}, we have
\begin{eqnarray*}
-\mathbb{E}[\Delta x_{i^{min}}(t)] &=& \frac{\sigma x_{i^{min}}(t)}{\omega-\alpha}\left[1+\frac{\omega-\sum_{l\in{\cal C}-{\cal C}_{\epsilon_U}(t)} \hspace{-0.1cm}x_l(t)\pi_l(t)}{\omega-\alpha}\right]
\left[\sum_{l\in{\cal C}-{\cal C}_{\epsilon_U}(t)} \hspace{-0.3cm}x_l(t)\pi_l(t)-\sum_{l\in{\cal C}-{\cal C}_{\epsilon_U}(t)} \hspace{-0.3cm}x_l(t)\pi_{i^{min}}(t)\right] \\
&>& \frac{\sigma x_{i^{min}}(t)}{\omega-\alpha}\frac{\epsilon_U}{N}.
\end{eqnarray*}

It then follows that
\begin{eqnarray*}
\pi_{i^{min}}(t+1)-\pi_{i^{min}}(t) &=& \frac{\mu_{i^{min}}}{Nx_{i^{min}}(t+1)}-\frac{\mu_{i^{min}}}{Nx_{i^{min}}(t)} = \frac{1}{N}\left[\frac{\mu_{i^{min}}}{x_{i^{min}}(t)+\Delta x_{i^{min}}(t)}-\frac{\mu_{i^{min}}}{x_{i^{min}}(t)}\right] \\
                                    & > & \frac{\mu_{i^{min}}[-\Delta x_{i^{min}}(t)]}{N[x_{i^{min}}(t)]^2} >\frac{\mu_{i^{min}}\sigma\epsilon_U}{(\omega-\alpha)N^2x_{i^{min}}(t)}>\frac{\mu_{i^{min}}\sigma\epsilon_U}{N^2}.
\end{eqnarray*}

Given that $\pi_i\in[0,1], \forall i\in{\cal C}$ and let $\mu_{min}\triangleq \min_{i \in {\cal C}} \mu_i$, it holds that after at most $O(\frac{N^2}{\mu_{min}\sigma\epsilon_U})$ iterations, DISAP converges to an imitation-stable equilibrium where $\pi_{i^{max}}-\pi_{i^{min}}\le\epsilon_U$.

It then can be shown in the same way as in the proof of Theorem~\ref{th:convergence_simple} that the imitation-stable equilibrium is an $\epsilon$-Nash of $G$ with $\epsilon=2\epsilon_U$.
\end{proof}

\begin{proof}[\textbf{Proof of theorem~\ref{th:finitePop}}]

We prove the statement for $t=2$. The case for $t\ge 3$ is analogous to \cite{Schlag96}, which can be shown by induction and is therefore omitted.


Define the random variable $w_i^j(c)$ such that
\begin{equation*}
w_i^j(c)=
\begin{cases}
1 & \text{if SU $c$ is on channel $j$ at iteration $t=1$} \\
  & \text{and migrates to channel $i$ at $t=2$} \\
0 & \text{otherwise}
\end{cases}.
\end{equation*}
\\
{\bf Proportional imitation}: If $j\neq s_c(1)$, it holds that $w_i^j(c)=0$, otherwise, $c$ imitates a SU that was using channel $k$ at $t=0$ and currently ($t=1$) on the same channel as $c$ ($s_c(1)$) with probability
$\frac{n_{s_c(1)}^k}{n_{s_c(1)}}$ and migrates to channel $i$ with probability $F_{s_c(0),k}^{i}$. Note that we allow for self-imitation in our algorithm.
We thus have:
\begin{eqnarray*}
\mathbb{P}[w_i^j(c)=1] =
\begin{cases}
0 & \mbox{ if } j\neq s_c(1) \\
\displaystyle \sum_{k\in \mathcal{C}} \frac{n_{s_c(1)}^{k}}{n_{s_c(1)}}F_{s_c(0),k}^{i}  & \mbox{ otherwise}
\end{cases}.
\end{eqnarray*}

We can now derive the population proportions at iteration $t=2$ as:
\begin{equation*}
p_i^j(2)=\frac{1}{N}\sum_{c\in {\cal N}}w_i^j(c) \quad \forall i, j\in \mathcal{C}.
\end{equation*}
The expectations of these proportions can now be written as (using the Kronecker delta $\delta_{i,j}$):

\begin{eqnarray*}
\mathbb{E}[p_i^j(2)] &=& \frac{1}{N}\sum_{c\in \mathcal{N}}\mathbb{P}[w_i^j(c)=1] = \frac{1}{N}\sum_{c\in \mathcal{N},k\in \mathcal{C}} \frac{n_{s_c(1)}^k(1) F_{s_c(0),k}^i\delta_{j,s_c(1)}}{n_{s_c(1)}(1)} \\
&=&\frac{1}{N}\sum_{h,l,k\in \mathcal{C}} \frac{n_h^l(1)n_h^k(1)F_{l,k}^{i}\delta_{j,h}}{n_h(1)} = \frac{1}{N}\sum_{l,k\in \mathcal{C}} \frac{n_j^l(1)n_j^k(1)F_{l,k}^{i}}{n_j(1)} = \sum_{l,k\in \mathcal{C}} \frac{\tilde{x}_j^l(1)\tilde{x}_j^k(1)}{\tilde{x}_j(1)}F_{l,k}^{i}.
\end{eqnarray*}

It follows that
\begin{eqnarray*}
\mathbb{E}[p_i(2)]&=&\sum_{j \in \mathcal{C}}\mathbb{E}[p_i^j(2)] = \sum_{j,l,k \in \mathcal{C}} \frac{\tilde{x}_j^l(1)\tilde{x}_j^k(1)}{\tilde{x}_j(1)}F_{l,k}^{i}.
\end{eqnarray*}


As $w_i^j(c)$ and $w_i^j(d)$ are independent random variables for $c\neq d$ and since the variance of $w_i^j(c)$ is less than $1$, the variance of $p_i^j(2)$ and $p_i(2)$ for any $i,j\in{\cal C}$ are less than $1/N$ and $C/N$, respectively. It then follows the Bienaym\'e-Chebychev inequality that
\begin{equation*}
\forall i\in{\cal C}, \mathbb{P}[  \{ | p_i(2)-\mathbb{E}[p_i(2)] | > \delta  \} ] < \frac{C}{(N\delta)^2}.
\end{equation*}

Choosing $N_{0}$ such that $\frac{C}{(N_0\delta)^2}<\epsilon$ concludes the proof for $t=2$.
The proof can then be induced to any $t$ as in \cite{Schlag96}.

{\bf Double imitation}: If $j\neq s_c(1)$, it holds that $w_i^j(c)=0$, otherwise, $c$ imitates two SUs that were using respectively channel $k$ and channel $z$ at $t=0$ and currently ($t=1$) on the same channel as $c$ ($s_c(1)$) with
probability $\frac{n_{s_c(1)}^k}{N_{s_c(1)}}\frac{n_{s_c(1)}^{z}}{n_{s_c(1)}}$ and migrates to channel $i$ with probability $F_{s_c(0),\{k,z\}}^{i}$.\\
The proof follows in the steps of the proportional imitation and only the main passages will be sketched out. We get:
\begin{eqnarray*}
\mathbb{P}[w_i^j(c)=1] =
\begin{cases}
0 & \mbox{ if } j\neq s_c(1) \\
\displaystyle \sum_{k,z\in \mathcal{C}} \frac{n_{s_c(1)}^{k}}{n_{s_c(1)}}\frac{n_{s_c(1)}^{z}}{n_{s_c(1)}}F_{s_c(0),\{k,z\}}^{i}  & \mbox{ otherwise}
\end{cases}.
\end{eqnarray*}

We then derive the type proportions expectations:
\begin{eqnarray*}
\mathbb{E}[p_i^j(2)] &=& \frac{1}{N}\sum_{c\in \mathcal{N}}\mathbb{P}[w_i^j(c)=1] = \frac{1}{N}\sum_{c\in \mathcal{N},k,z\in \mathcal{C}} \frac{n_{s_c(1)}^k(1) }{n_{s_c(1)}(1)}\frac{n_{s_c(1)}^{z}(1)}{n_{s_c(1)}(1)}F_{s_c(0),\{k,z\}}^i\delta_{j,s_c(1)} \\
&=& \sum_{l,k,z\in \mathcal{C}} \frac{\tilde{x}_j^l(1)\tilde{x}_j^k(1)\tilde{x}_j^z(1)}{[\tilde{x}_j(1)]^2}F_{l,\{k,z\}}^{i}.
\end{eqnarray*}

It follows that:
\begin{eqnarray*}
\mathbb{E}[p_i(2)]&=&\sum_{j \in \mathcal{C}}\mathbb{E}[p_i^j(2)] = \sum_{j,l,k,z \in \mathcal{C}} \frac{\tilde{x}_j^l(1)\tilde{x}_j^k(1)\tilde{x}_j^z(1)}{[\tilde{x}_j(1)]^2}F_{l,\{k,z\}}^{i}.
\end{eqnarray*}
The rest of the proof for the double imitation follows the same way as that of proportional imitation.
\end{proof}

\begin{proof}[\textbf{Proof of theorem~\ref{th:convergence_pi_2}}]

We prove the convergence of \eqref{eq:system_pi} by showing that the mapping described by \eqref{eq:system_pi} is a contraction. A contraction mapping is defined \cite{Abraham88} as follows: let
$(X,d)$ be a metric space, $f$: $X \rightarrow X$ is a contraction
if there exists a constant $k\in[0,1)$ such that $\forall x, y \in
X$, $d(f(x), f(y)) \leq kd(x,y)$, where $d(x, y)= ||x-y|| =\max_i
|x_i - y_i|$. Such an $f$ is called a contraction and admits a unique fixed point, to which the mapping described by $f$ converges.

Noticing that
\begin{eqnarray*}
 d(f(x), f(y)) = ||f(x)-f(y)|| \leq
\left|\left|{\partial f \over
\partial x}\right|\right| \cdot||x-y||= \left|\left|{\partial f \over \partial x}\right|\right| d(x,y),
\end{eqnarray*}
it suffices to show that the Jacobian $\displaystyle \left|\left|{\partial f \over
\partial x}\right|\right| \leq k$. In our case, it suffices to show that $||J||_{\infty} \leq k$, where
$J=\{J_{ij}\}$ is the Jacobian of the mapping described by one of the equation in \eqref{eq:system_pi}, defined by
$\displaystyle J_{ij}={\partial x_i(t+1) \over
\partial x_j(t)}$.

Recall that $\pi_i=\frac{\mu_i}{Nx_i}$ and $\bar{\pi}=\sum_l \frac{\mu_l}{N}$, \eqref{eq:system_pi} can be rewritten as
\begin{eqnarray*}
x_i(u)=x_i(u-1)+\sigma\left[\frac{\mu_i}{N}-x_i(u-1)\sum_l \frac{\mu_l}{N}\right].
\end{eqnarray*}

It follows that
\vspace{-0.5cm}
\begin{eqnarray*}
J_{ij} =
\begin{cases}
\displaystyle 1-\sum_l \frac{\mu_l}{N} & j=i \\
0 & \text{otherwise}
\end{cases}.
\end{eqnarray*}

Hence
\vspace{-0.5cm}
\begin{eqnarray*}
||J||_{\infty}=\max_{i\in{\cal N}} \sum_{j\in{\cal N}} \left|J_{ij}\right|=1-\sum_l \frac{\mu_l}{N}<1,
\end{eqnarray*}
which shows that the mapping described by \eqref{eq:system_pi} is a contraction. It is further easy to check that the fixed point of \eqref{eq:system_pi} is $\{x^*=\frac{\mu_i}{\sum_{l\in{\cal N}} \mu_l}\}$, which is also the unique NE of $G$.
\end{proof}

\newpage

\begin{algorithm}
\caption{PISAP: executed at each SU $j$}
\begin{algorithmic}[1]
\STATE \textbf{Initialization}: set the imitation factor $\sigma$ and the imitation threshold $\epsilon_U$
\STATE For the first iteration $t=1$, randomly choose a channel to stay
\WHILE{at each iteration $t\ge 2$}
    \STATE Randomly select a SU $j'$
    \IF{$U_j<U_{j'}-\epsilon_U$}
        \STATE Migrate to the channel $s_{j'}$ with probability $p=\sigma (U_{j'}-U_j)$
    \ENDIF
\ENDWHILE
\end{algorithmic}
\label{algo:pir}
\end{algorithm}

\begin{algorithm}
\caption{DISAP: executed at each SU $j$ for each iteration}
\begin{algorithmic}[1]
\STATE \textbf{Initialization}: set the two exogenous parameters $\omega$ and $\alpha$ such that the payoff of SUs falls into the interval $[\alpha,\omega]$, set the imitation factor $\sigma$ and the imitation threshold $\epsilon_U$
\STATE Randomly sample two SUs $j_1$ and $j_2$ who, at iteration $t-1$, were respectively on channel $i_1$ and $i_2$
\IF{$i_1=i_2$}
\STATE  $ \begin{array}{@{\hspace{-0mm}}r@{\;}l@{\hspace{0mm}}}
p_{j_1} = & \frac{\sigma }{2}\left[Q(U_{j_1})(U_{j_1}-U_{j})+Q(U_{j_2})(U_{j_2}-U_j)\right]^+ \mbox{ where } [A]^+ \mbox{ denotes } \max \{0, A\} \mbox{ and }\\ 
& Q(U_r)\triangleq\frac{1}{\omega - \alpha}\left[2 - \frac{U_r-\alpha}{\omega - \alpha}\right]\\
	\end{array}
$ \hfill \phantom{.}


$p_{j_2}=0$

\ELSIF{$i_1=i$}
\STATE $p_{j_1}=0$\\
$p_{j_2}=\frac{\sigma }{4}\left[Q(U_{j_1})(U_{j_2}-U_{j_1})+Q(U_j)(U_{j_2}-U_{j_1})\right]^+$
\ELSE
\STATE $p_{j_1} =\frac{\sigma }{2}\left[Q(U_j)(U_{j_1}-U_{j_2})+Q(U_{j_2})(U_{j_1}-U_j)\right]^+$\\
$p_{j_2}=\frac{\sigma }{2}\left[Q(U_{j_1})(U_{j_2}-U_j)+Q(U_{j_2})(U_{j_1}-U_j)\right]^+-p_{j_1}$
\ENDIF
\STATE Switch to channel $i_1$ with probability $p_{j_1}$ if $U_j<U_{j_1}-\epsilon_U$, switch to channel $i_2$ with probability $p_{j_2}$ if $U_j<U_{j_2}-\epsilon_U$
\end{algorithmic}
\label{algo:apir}
\end{algorithm}

\begin{algorithm}
\caption{Imitation-based Spectrum Access Policy under Channel Constraint: executed at each SU $j$}
\begin{algorithmic}[1]
\STATE \textbf{Initialization:} set the imitation factor $\sigma$ and the imitation threshold $\epsilon_U$
\STATE Randomly choose a channel for the first two iterations $t=0,1$
\WHILE{for each iteration $t\ge 2$}
        \STATE Perform imitation in PISAP or DISAP on the same channel
        \STATE $t \leftarrow t+1$:
\ENDWHILE
\end{algorithmic}
\label{algo:generic_imitation}
\end{algorithm}

\begin{figure}
 \centering
\includegraphics[width=0.49\linewidth]{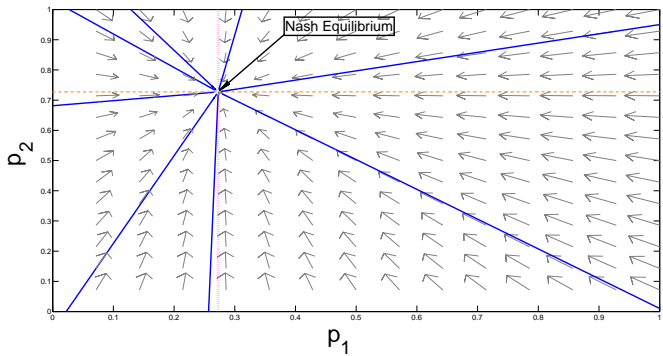}
\caption{Replicators and aggregate monotone dynamics generate a similar phase plane. This is the case of a 2-strategies game with $N=50$ SUs and $\mu{}=[0.3~0.8]$. As investigated in
Section~\ref{sec:imitation} and Section~\ref{sec:models}, the system has a unique NE, to which all trajectories (solid lines) converge exponentially}
\label{fig:phasePlane}
\end{figure}

\begin{figure*} 
\begin{minipage}[r]{0.49\linewidth}
\includegraphics[width=\linewidth]{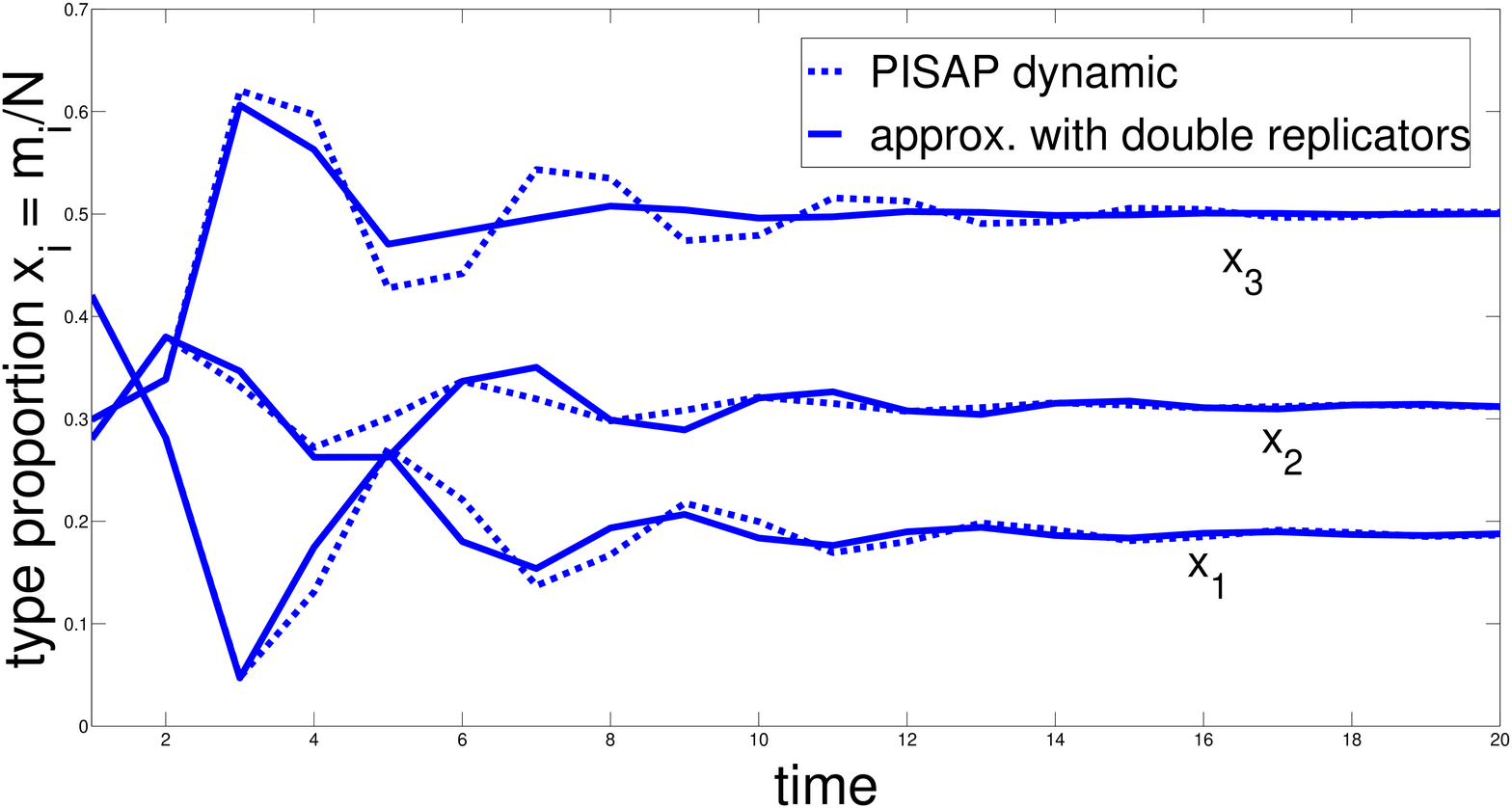}
\caption{PISAP dynamic and its approximation by double replicator dynamic.}
\label{fig:approximation}
\end{minipage} \hfill
\begin{minipage}[c]{0.49\linewidth}
\includegraphics[width=\linewidth]{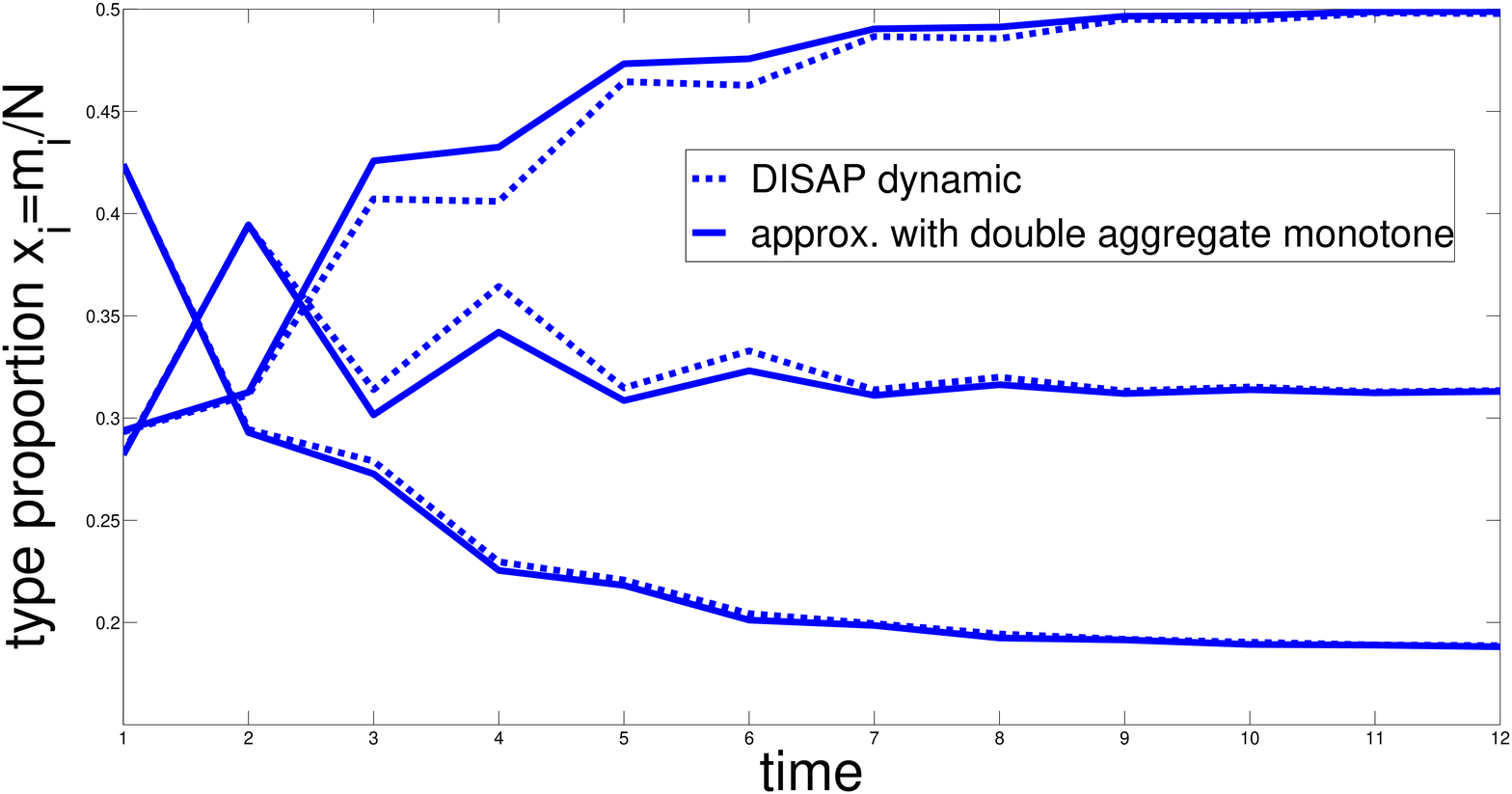}
\caption{DISAP dynamic and its approximation by double aggregate monotone dynamic.}
\label{fig:DIvsAPPROX}
\end{minipage}\hfill
\vspace{0.5cm}
\end{figure*}

\begin{figure*}
\begin{minipage}[r]{0.49\linewidth}
\includegraphics[width=\linewidth]{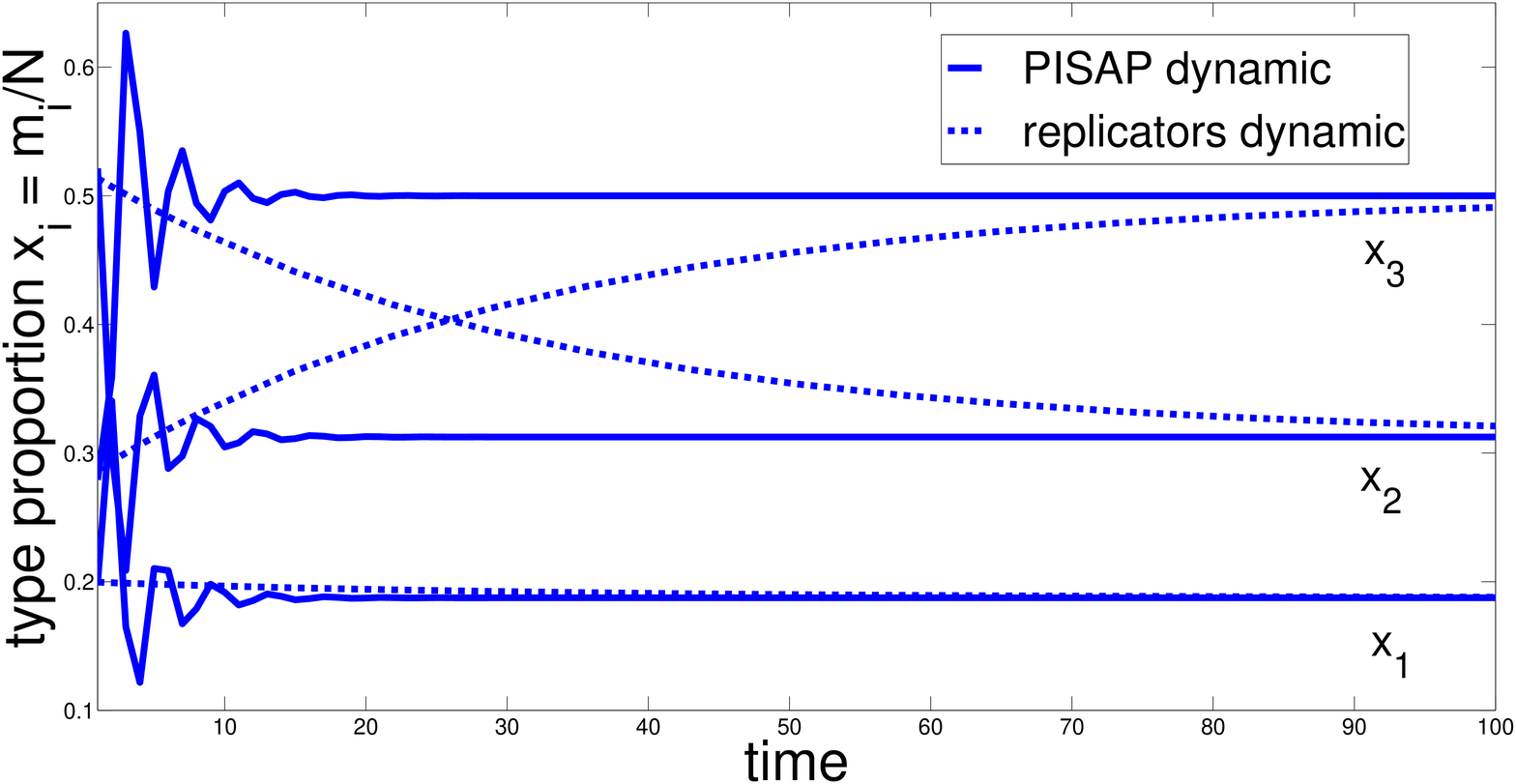}
\caption{PISAP dynamic with channel constraint and replicator dynamic without channel constraint}
\label{fig:oursVSreplicators}
\end{minipage} \hfill
\begin{minipage}[c]{0.49\linewidth}
\includegraphics[width=\linewidth]{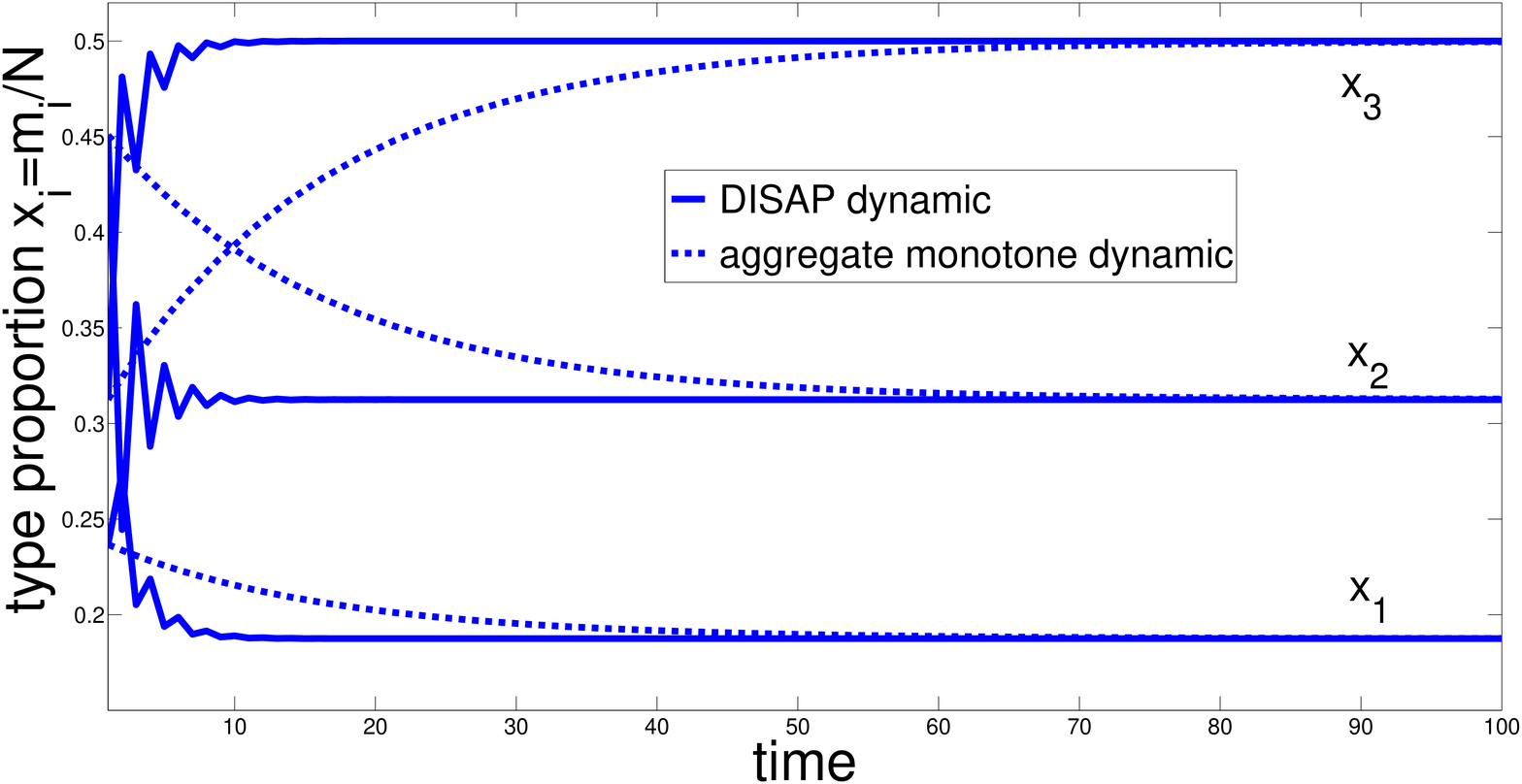}
\caption{DISAP dynamic with channel constraint and aggregate monotone dynamic without channel constraint}
\label{fig:oursVSaggregate}
\end{minipage}\hfill
\vspace{0.5cm}
\end{figure*}

\begin{figure}
\centering
\includegraphics[width=0.49\linewidth]{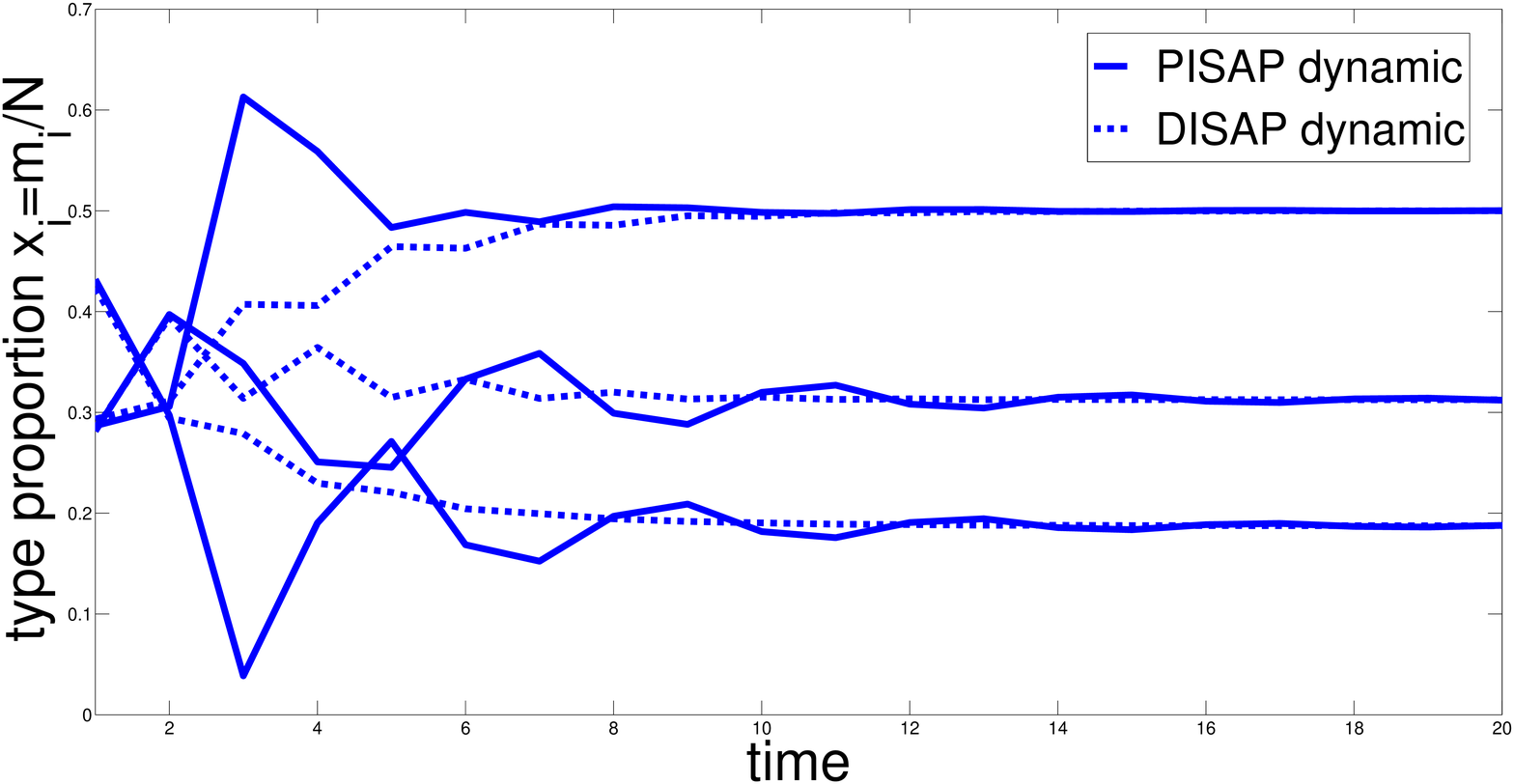}
\caption{DISAP dynamic compared to PISAP dynamic \textbf{with} channel constraint.}
\label{fig:DIvsPIR}
\end{figure}

\begin{figure*}
\begin{minipage}[r]{0.49\linewidth}
\includegraphics[width=\linewidth]{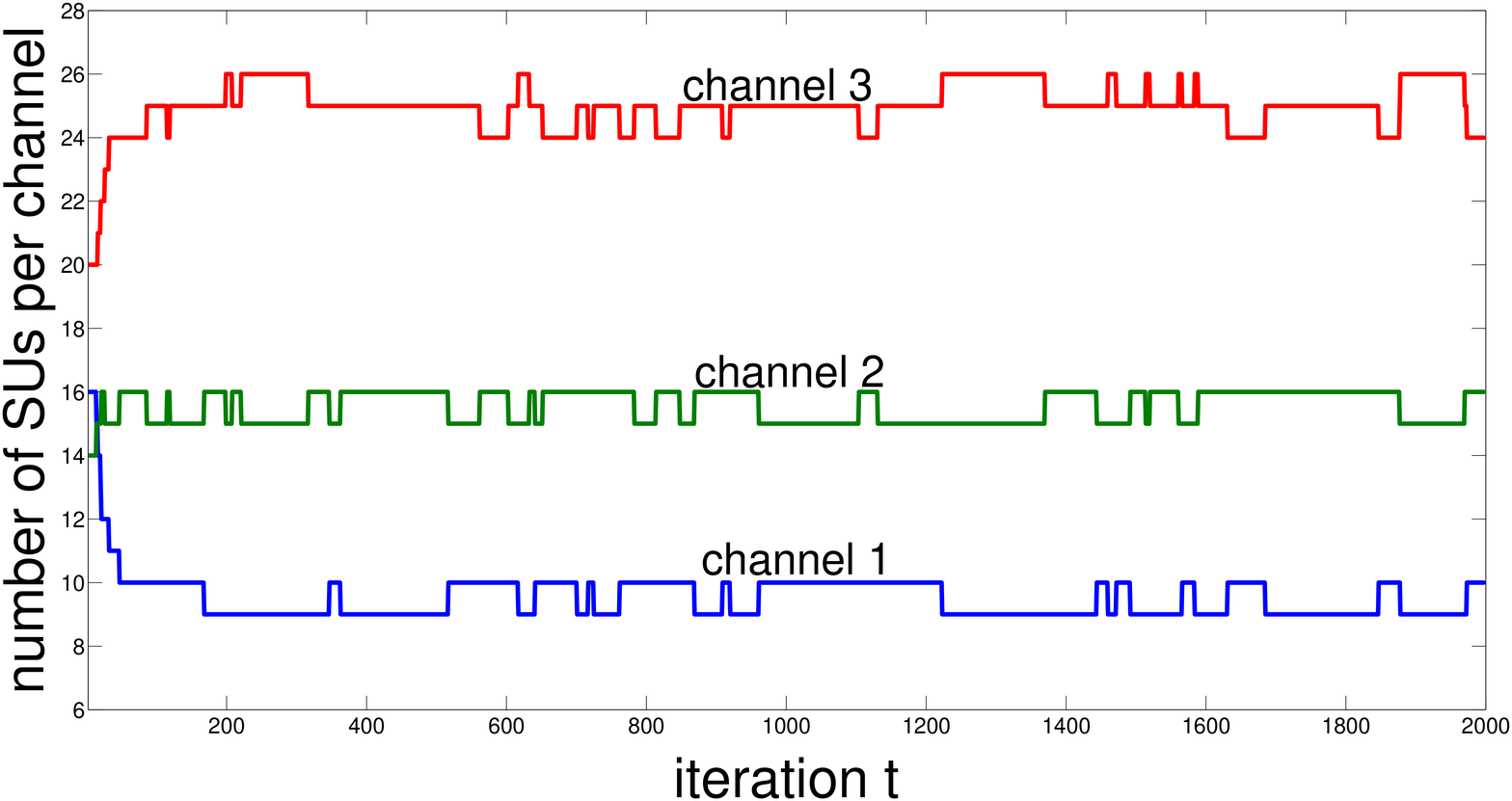}
\caption{PISAP: number of SUs per channel as a function of time \textbf{without} channel constraint}
\label{fig:pirSuxCHwoCC}
\end{minipage} \hfill
\begin{minipage}[c]{0.49\linewidth}
\includegraphics[width=\linewidth]{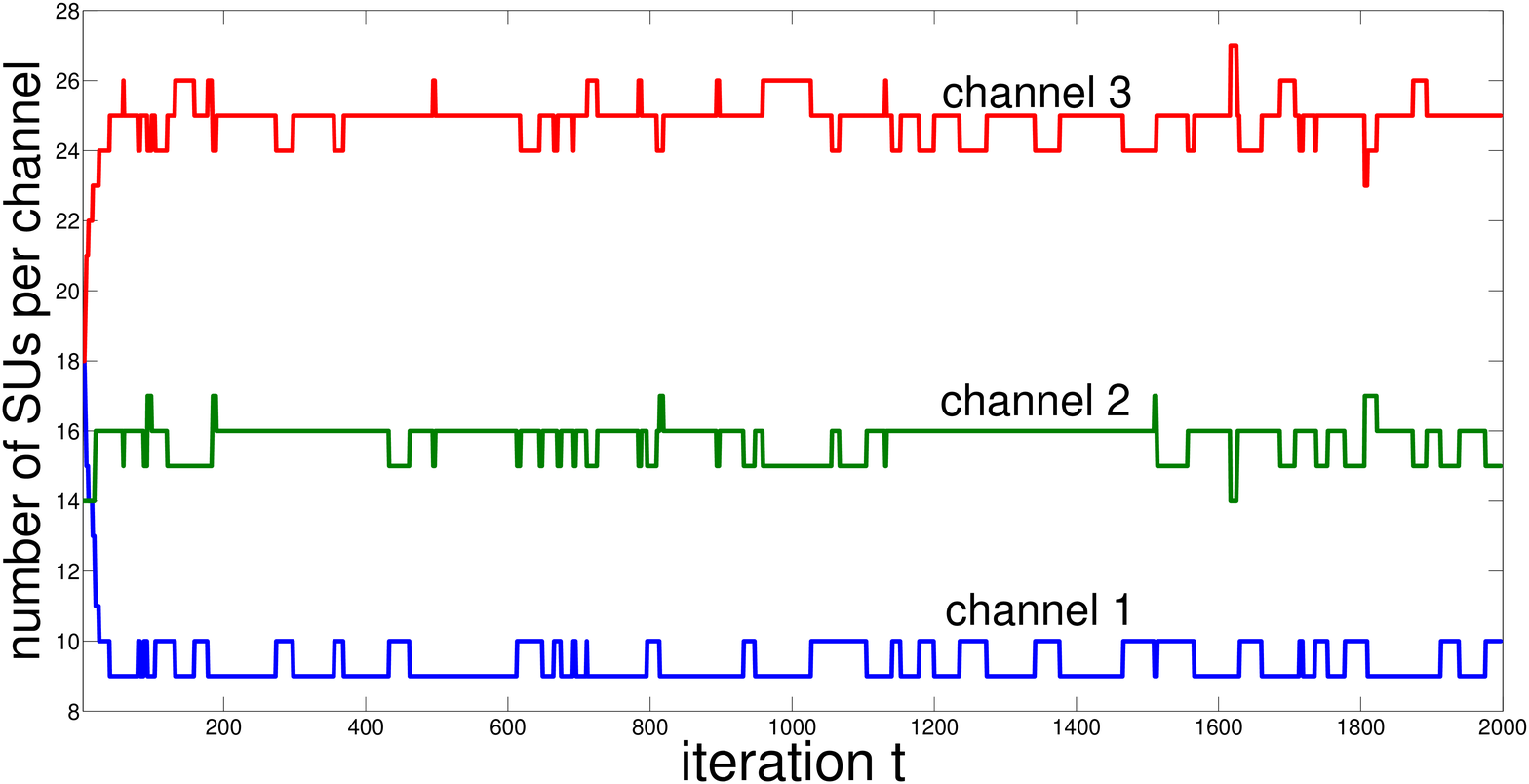}
\caption{DISAP: number of SUs per channel as a function of time \textbf{without} channel constraint}
\label{fig:diSuxCHwoCC}
\end{minipage}\hfill
\end{figure*}

\begin{figure*}
\begin{minipage}[r]{0.49\linewidth}
\includegraphics[width=\linewidth]{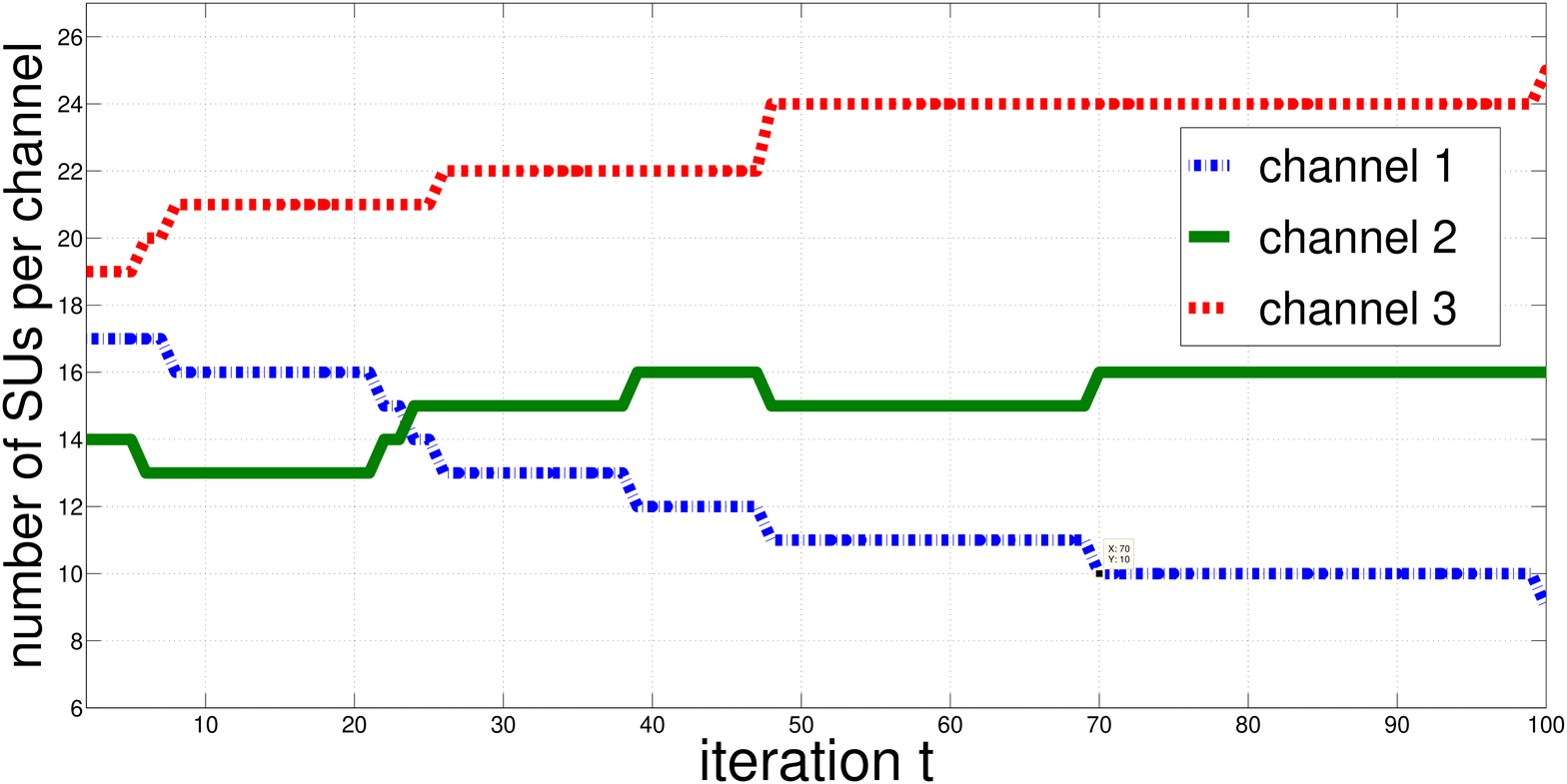}
\caption{PISAP: focus on the convergence phase \textbf{without} channel constraint}
\label{fig:PIR_convergencePH}
\end{minipage} \hfill
\begin{minipage}[c]{0.49\linewidth}
\includegraphics[width=\linewidth]{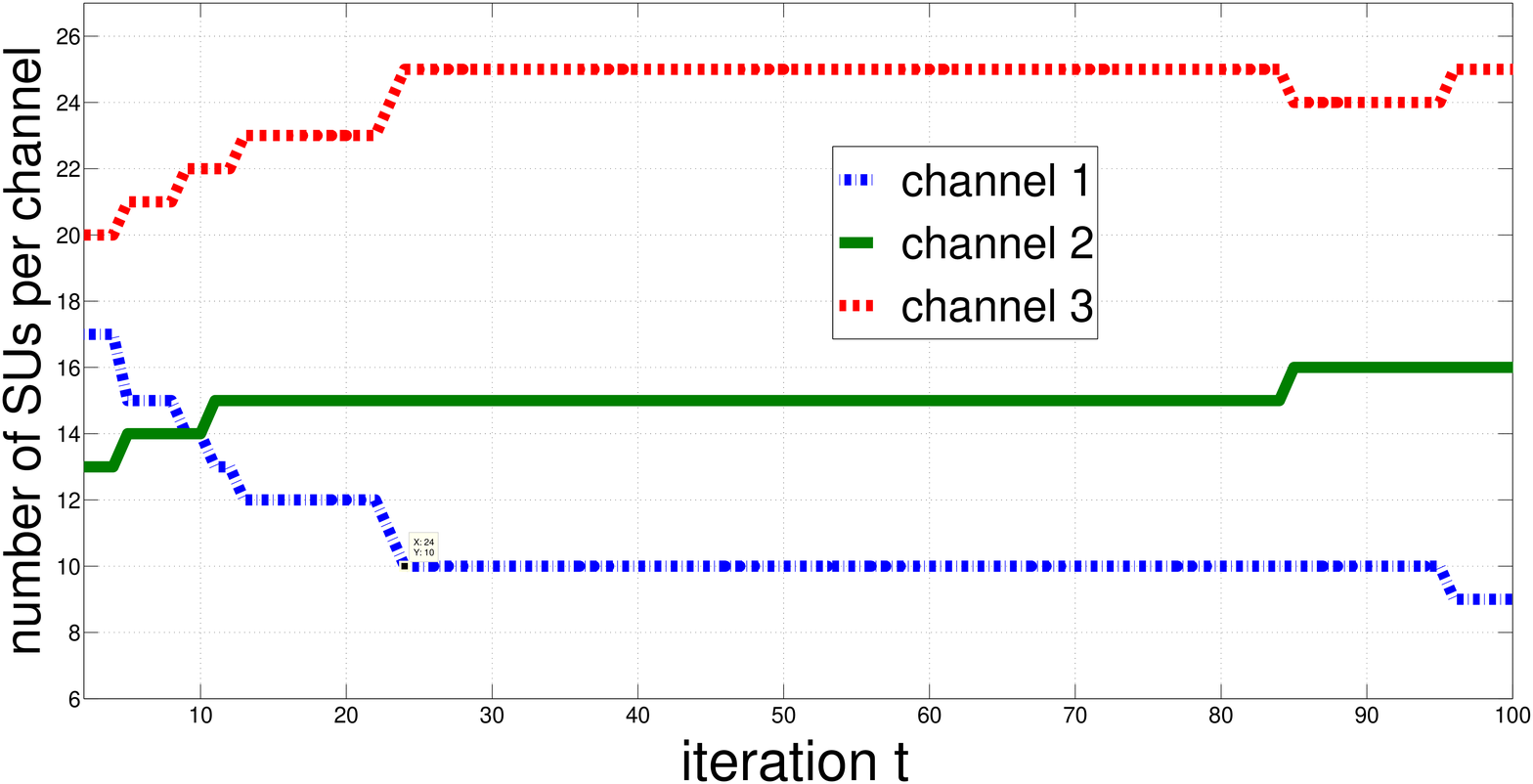}
\caption{DISAP: focus on the convergence phase \textbf{without} channel constraint}
\label{fig:DI_convergencePH}
\end{minipage}\hfill
\vspace{0.5cm}
\end{figure*}

\begin{figure*} 
\begin{minipage}[r]{0.49\linewidth}
\includegraphics[width=\linewidth]{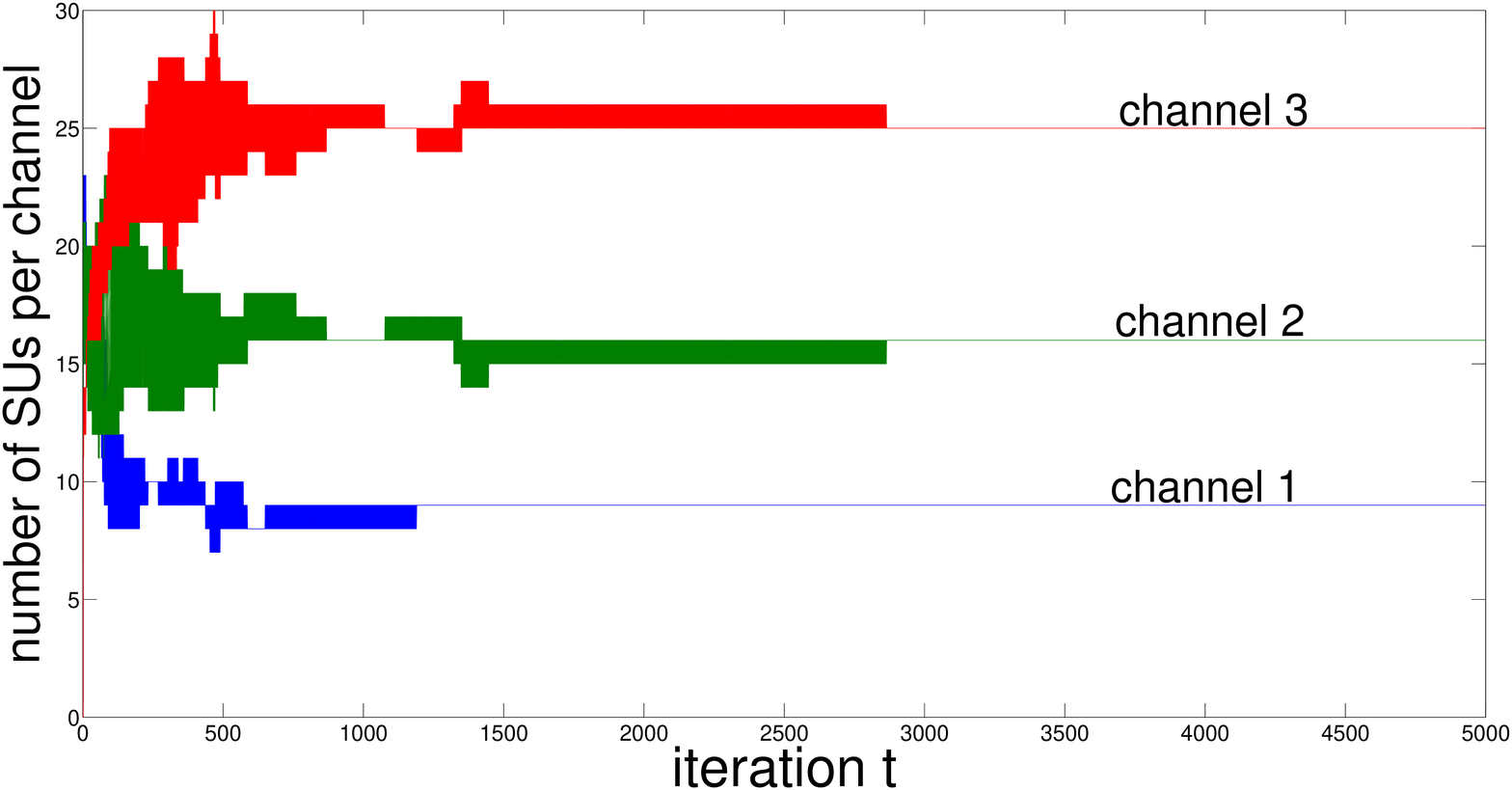}
\caption{PISAP: number of SUs per channel as a function of time \textbf{with} channel constraint}
\label{fig:PIR_SIneighbors_NOlearning}
\end{minipage} \hfill
\begin{minipage}[c]{0.49\linewidth}
\includegraphics[width=\linewidth]{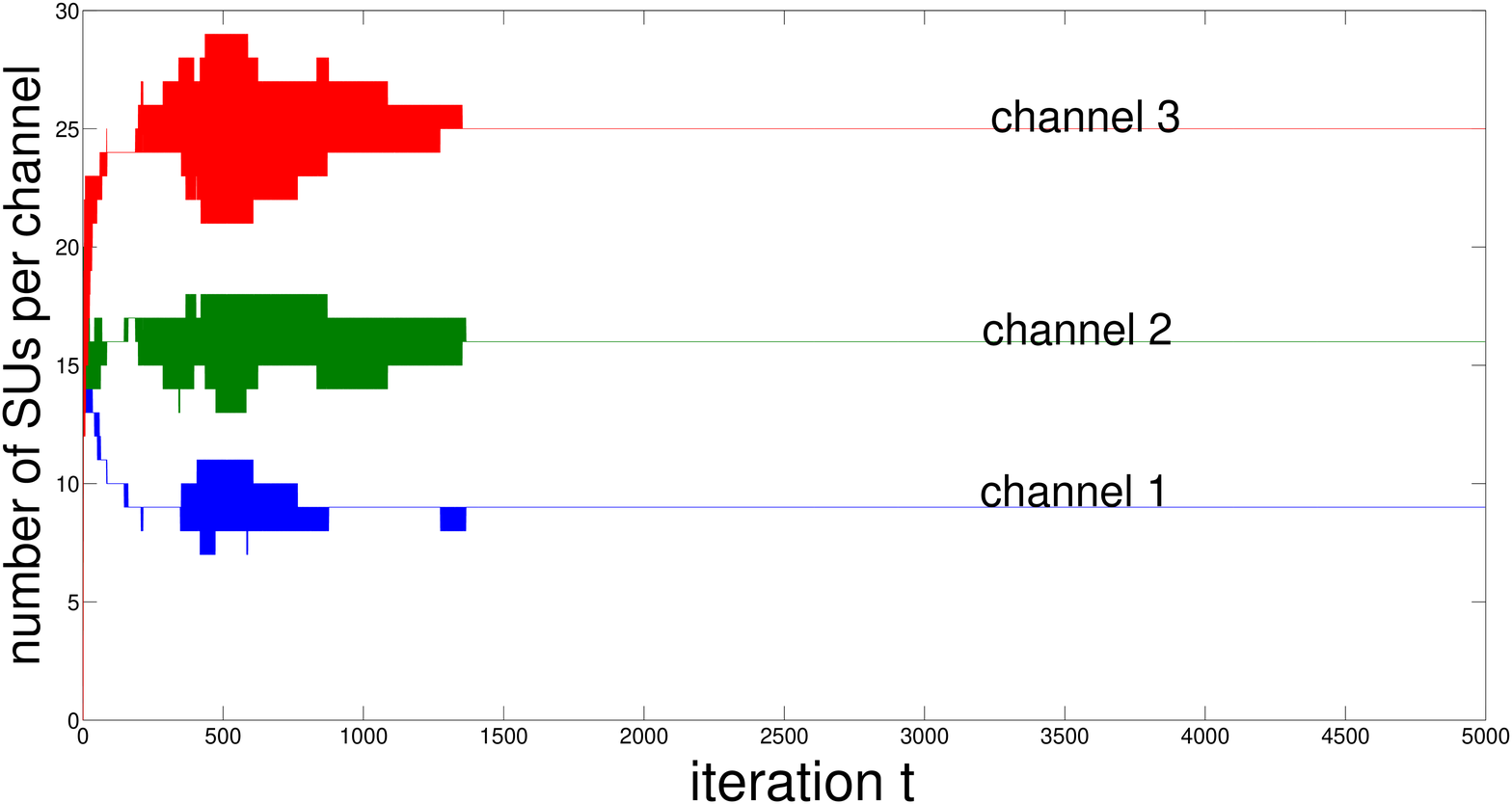}
\caption{DISAP: number of SUs per channel as a function of time \textbf{with} channel constraint}
\label{fig:DI_SIneighbors_NOlearning}
\end{minipage}\hfill
\vspace{0.5cm}
\end{figure*}

\begin{figure}[tbp]
 \centering
\includegraphics[width=0.49\linewidth]{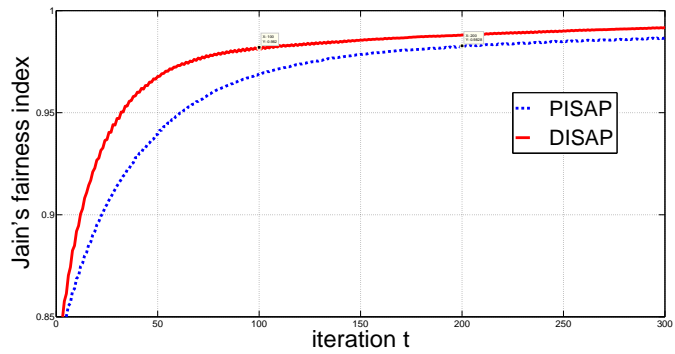}
\vspace{-0.5cm}
\caption{Jain's fairness index of the system with channel constraint as a function of time}
\label{fig:jain}
\end{figure}

\end{document}